\documentclass[review]{elsarticle}
\usepackage[letterpaper,top=2cm,bottom=2cm,left=3cm,right=3cm,marginparwidth=1.75cm]{geometry}
\usepackage{lineno,hyperref,mathtools}

\journal{Journal of \LaTeX\ Templates}

\usepackage{amssymb}
\usepackage{bm}
\usepackage{amssymb}
\usepackage{amsthm}
\usepackage{multirow,booktabs}
\usepackage{cases}
\usepackage{threeparttable}
\usepackage{color}

\newtheorem{theorem}{Theorem}
\newtheorem{proposition}{Proposition}
\newtheorem{remark}{Remark}

\newtheorem{lemma}{Lemma}
\newtheorem{corollary}{Corollary}

\newcommand{\Hull}{{\mathrm{Hull}}}
\newcommand{\C}{{\mathcal{C}}}
\newcommand{\F}{{\mathbb{F}}}

\begin{document}

\begin{frontmatter}

\title{Three classes of propagation rules for generalized Reed-Solomon codes and their applications to EAQECCs}
\tnotetext[mytitlenote]{This research work is supported by the National Natural Science Foundation of China under Grant Nos. U21A20428 and 12171134.}

\author[mymainaddress]{Ruhao Wan}
\ead{wanruhao98@163.com}

\author[mymainaddress]{Shixin Zhu\corref{mycorrespondingauthor}}
\cortext[mycorrespondingauthor]{Corresponding author}
\ead{zhushixinmath@hfut.edu.cn}

\address[mymainaddress]{School of Mathematics, HeFei University of Technology, Hefei 230601, China}

\begin{abstract}
In this paper, we study the Hermitian hulls of generalized Reed-Solomon (GRS) codes over finite fields.
For a given class of GRS codes,
by extending the length, increasing the dimension, and extending the length and increasing the dimension at the same time, we obtain three classes of GRS codes with Hermitian hulls of arbitrary dimensions.
Furthermore,
based on some known $q^2$-ary Hermitian self-orthogonal GRS codes with dimension $q-1$, we obtain several classes of $q^2$-ary maximum distance separable (MDS) codes with Hermitian hulls of arbitrary dimensions.
It is worth noting that the dimension of these MDS codes can be taken from $q$ to $\frac{n}{2}$, and the parameters of these MDS codes can be more flexible by propagation rules.
As an application,
we derive three new propagation rules for MDS entanglement-assisted quantum error correction codes (EAQECCs) constructed from GRS codes.
Then, from the presently known GRS codes with Hermitian hulls, we can directly obtain many MDS EAQECCs with more flexible parameters.
Finally, we present several new classes of (MDS) EAQECCs with flexible parameters, and the distance of these codes can be taken from $q+1$ to $\frac{n+2}{2}$.
\end{abstract}

\begin{keyword}
Hermitian hulls \sep GRS codes \sep propagation rules \sep EAQECCs
\end{keyword}

\end{frontmatter}

\section{Introduction}\label{sec1}

Let $\F_{q}$ be the finite field with $q$ elements, where $q=p^m$ is a prime power.
Let $\F_{q}^n$ be the $n$-dimensional vector space over $\F_{q}$, where $n$ is a positive integer.
An $[n,k,d]_{q}$ linear code $\C$ over $\F_{q}$ is a $k$-dimensional subspace of $\F_{q}^n$ and minimum distance $d$.
A linear code is called an MDS code if $d=n-k+1$.
For a linear code $\C$,
the hull of $\C$ is defined as $\Hull(\C)=\C\cap \C^{\bot}$,
where $\C^\bot$ be the dual code of $\C$ with respect to certain inner product.
Hulls of linear codes were found to have wide applications
in determining the complexity of algorithms for computing the automorphism group of a linear code (see \cite{RefJ (1982) app J}).
Moreover, hull plays an important role in checking permutation equivalence of two linear codes
(see \cite{RefJ (1991) app J.S,RefJ (2000) app N}).

There are two special cases of hulls that are of great interest.
One is that $\Hull(\C)=\{0\}$, in which the code $\C$ is just the linear complementary dual (LCD) code.
LCD codes were first introduced by Massey in \cite{RefT (1992) LCD Ma}.
Later, LCD codes were widely used to protect against side-channel attacks (SCAs) and fault injection attacks (FIAs) (see \cite{RefT (2014) LCD CC,RefJ (2016) app C}).
Since then,
the study of LCD codes has received much attention (see \cite{RefJ (2018) Carlet equivalent}-\cite{RefJ (2017) LCD Li} and the references therein).
Another special case is $\Hull(\C)=\C$, in which case the code $\C$ is called self-orthogonal.
In \cite{RefJ (1996) introduct 1}-\cite{RefJ (2001) introduct 3}, the authors
established the fundamentals for constructing a class of quantum codes named stabilizer codes
using classical linear codes.
Then
a stabilizer code can be yielded if there exists a classical linear codes with certain self-orthogonality.

For the case of arbitrary hull dimensions, Brun et al. \cite{RefJ (2006) bound} introduced a class of quantum codes called EAQECCs.
We can find that quantum stabilizer codes as a special case of EAQECCs.
In other words, one can construct a EAQECC by a classical linear code without self-orthogonality.
However,
it is not easy to determine the number of shared pairs that required to construct an EAQECC.
Fortunately, Guenda et al. \cite{RefJ (2018) K good} found that this number can be obtained from the dimension of the hull of a classical linear code.
Thus, we can obtain an EAQECC from a classical linear code with a determined dimensional hull.

\subsection{Recent related works}

Recently, the construction of (MDS) EAQECCs has attracted the attention of many
researchers and the related results have been gradually improved.
More precisely, all $q$-ary MDS EAQECCs of length $n\leq q+1$ have been constructed in \cite{RefJ (2020) K q+1}.
Since then, researchers have focused on the construction of $q$-ary MDS EAQECCs with length $n>q+1$, in particular, the construction of $q$-ary MDS EAQECCs with distances $d>q+1$.
By using cyclic codes, constacyclic codes and GRS codes, many (MDS) EAQECCs with good parameters have been constructed (see \cite{RefJ (2006) Ketkar}-\cite{RefJ (2024) Anderson} and the references therein).
We summarise the relevant results as follows.

\begin{itemize}
\item
In \cite{RefJ (2022) Chen IEEE} and \cite{RefJ (2023) Chen}, Chen studied the
 Hull-Variation problem of equivalent linear codes.
 From this, we know that
  if there exists a $q^2$-ary linear code $\C$ with $l$-dimension Hermitian hull, where $q>2$,
  then there exists an equivalent code $\C'$ of $\C$ with $s$-dimensional Hermitian hull for each $0\leq s\leq l$.

\item
In \cite{RefJ (2024) Anderson},
Anderson et al. extended Chen's results \cite{RefJ (2022) Chen IEEE,RefJ (2023) Chen} to the relative hulls of two codes.
From this, we know that
when $q>2$, the relative hull dimension
can be repeatedly reduced by one, down to a certain bound,
by replacing either of the two codes with an equivalent one.
Moreover, under certain conditions, the relative hull dimension
can be increased by one via equivalent codes.

 \item
In \cite{RefJ (2022) Luo}, Luo et al. obtained three classes of propagation rules by considering equivalent codes, extending the length, and extending the length and increasing the dimension at the same time.
Then, Luo et al. obtained some EAQECCs with good parameters.

 \item
In \cite{RefJ (2022) Luo rule}, Luo et al. proposed two classes of propagation rules on quantum codes based on punctured and shortened classical codes.
From this, many known MDS EAQECCs can have more flexible parameters.
Moreover, in \cite{RefJ (2024) Luo MDS}, Luo et al. gave some GRS codes whose Hermitian hulls are MDS.
\end{itemize}

\subsection{Our motivations and contributions}

Our main motivations can be summarized as follows:

\begin{itemize}
\item
Although, in \cite{RefJ (2022) Luo}, Luo et al. have obtained the propagation rule by
extending the length, and extending the length and increasing the dimension at the same time, we find that for MDS codes using this propagation rule to obtain codes is not always an MDS code. Therefore, finding new propagation rules remains an interesting problem.

\item
To the best of our knowledge, when $n>q+1$, there are currently very few
MDS codes with
parameters $[n,k]_{q^2}$ for any $q<k\leq \frac{n}{2}$ such that their Hermitian
hull dimensions can be computed explicitly.
Therefore, as mentioned in \cite{RefJ (2020) Fang hull}, constructing $q^2$-ary MDS codes of length $n>q+1$, dimension $q<k\leq \frac{n}{2}$ and determining the dimensions of their Hermitian hulls is a valuable work.

\item
Up to now, the construction of MDS EAQECCs with distances greater than $q+1$ is still an
 active topic.
However,
there are only a few special length MDS EAQECCs whose distances can be greater than $q+1$,
 and very few MDS EAQECCs whose distances can be taken from $q+1$ to $\frac{n+2}{2}$.
Thus it makes sense to construct new MDS EAQECCs with distances greater than $q+1$.
\end{itemize}

The main contributions of this paper can be summarized as follows:

\begin{itemize}
\item
For a given class of GRS codes,
by extending the length, increasing the dimension, and extending the length and increasing the dimension at the same time, we obtain three classes of GRS codes with determined Hermitian hull dimensions (see Propositions \ref{pro GRS 1}, \ref{pro 3}, \ref{pro 4}).
From this, we can obtain many direct results (see Corollaries \ref{cor GRS 1} - \ref{cor GALOIS}).

\item
Based on some known $q^2$-ary Hermitian self-orthogonal GRS codes with dimension $q-1$, by considering larger dimensions,
we obtain several classes of $q^2$-ary MDS codes with Hermitian hulls of arbitrary dimensions (see Theorems \ref{th MDS 1}, \ref{th MDS 2}, \ref{th MDS 3}).
It is worth noting that the dimension of these MDS codes can be taken from $q$ to $\frac{n}{2}$.

\item
As an application,
we derive three new propagation rules on MDS EAQECCs constructed from GRS codes (see Theorems \ref{th rule 1}, \ref{th rule 2}, \ref{th rule 3}).
From this, many currently known MDS EAQECCs can have more flexible parameters.

\item
Finally,
 we present several new classes of MDS EAQECCs whose distances can be taken from $q+1$ to $\frac{n+2}{2}$
 (see Theorems \ref{th q22}).
 By comparison, the MDS EAQECCs we obtained have new parameters (see Tables \ref{tab:MDS EAQECC}, \ref{tab:EAQECCs}).
\end{itemize}

\subsection{Organization of this paper}

This paper is organized as follows.
In Section \ref{sec2}, we introduce some basic notations and results on GRS codes and hulls.
In Section \ref{sec3}, we present three classes of MDS codes from GRS codes.
In Section \ref{sec MDS}, we obtain several classes of MDS codes with Hermitian hulls of arbitrary dimensions.
In Section \ref{sec4}, we give three new classes of propagation rules and construct some new MDS EAQECCs
with distances greater than $q+1$.
Finally, Section \ref{sec6} concludes this paper.

\section{Preliminaries}\label{sec2}

Now, we review some of the basic concepts and conclusions about dual codes, GRS codes and hulls.
Let $q=p^m$, where $p$ is prime. For any two vectors $\bm{x}=(x_1,x_2,\dots,x_n),\ \bm{y}=(y_1,y_2,\dots,y_n)\in \F_{q}^n$,
the $e$-Galois inner product \cite{RefJ (2017) Galois} of vectors $\bm{x}$ and $\bm{y}$ is defined as
$$\langle \bm{x},\bm{y} \rangle_{e}=\sum_{i=1}^{n}x_iy_i^{p^e},$$
where $0\leq e\leq m-1$.
It is easy to see that
$\langle -, -\rangle_{0}$ (resp. $\langle -,-\rangle_{\frac{m}{2}}$ with even $m$)
 is just the Euclidean (resp. Hermitian) inner product.
Hence, the $e$-Galois inner product generalizes the Euclidean inner product and the Hermitian inner product.
Let $\C$ be a linear code over $\F_q$, the $e$-Galois dual code of $\C$ is defined as
$$\C^{\perp_{e}}=\{\bm{x}\in \F_{q}^n:\langle \bm{y}, \bm{x}\rangle_{e}=0, \ for \ all\ \bm{y}\in \C \}.$$

We define the $e$-Galois hull of $\C$ as $\C\cap\C^{\perp_e}$ denoted by $\Hull_e(\C)$.
In particular, if $\Hull_e (\C)=\C$, then $\C$ is an $e$-Galois self-orthogonal code.
For an even integer $m$, the Hermitian dual code and Hermitian hull of $\C$ are defined as $\C^{\perp_H}=\C^{\perp_{\frac{m}{2}}}$ and $\Hull_H(\C)=\C\cap\C^{\perp_{\frac{m}{2}}}$, respectively.
Next, when we consider the Hermitian case, we assume that the base field is $\F_{q^2}$.
We fix some notations as follows for convenience.

\begin{itemize}
\item
For a vector $\bm{c}=(c_1,c_2,\dots,c_n)\in \F_{q}^n$,
we denote a vector $(c_1^i,c_2^i,\dots,c_n^i)\in \F_{q}^n$ by $\bm{c}^i$ for any integer $i$. Specially, $0^0=1$.
\item
For a matrix $A=(a_{ij})$ over $\F_{q}$,
we define matrices $(a_{ji})$, $(a_{ij}^{p^{m-e}})$ and $(a_{ji}^{p^{m-e}})$ as $A^T$, $A^{p^{m-e}}$ and $A^\ddagger$, respectively.
\item
Similarly, for a matrix $A=(a_{ij})$ over $\F_{q^2}$,
we define matrices $(a_{ji})$, $(a_{ij}^q)$ and $(a_{ji}^q)$ as $A^T$, $A^{q}$ and $A^\dagger$, respectively.
\item
 Suppose $O_{i\times j}$ is a zero matrix of size $i\times j$ and $E_k$ is a identity matrix of size $k\times k$.
  \item
  For $k\times n$ matrix $G$ with row vectors $\bm{x}_1,\bm{x}_2,\dots,\bm{x}_k\in \F_q^n$, write $G=(\bm{x}_1,\bm{x}_2,\dots,\bm{x}_k)^{\mathcal{T}}$.
\end{itemize}

 A \emph{monomial matrix} is a square matrix with exactly one nonzero entry in each row and each column and zeros elsewhere.
Based on the monomial matrix, equivalence relations for linear codes can be defined:
Let two linear codes $\C$ and $\C'$ with respective generator matrices $G$ and $G'$ be given.
The codes $\C$ and $\C'$ are (monomially) equivalent if there exists a monomial matrix $M$ such that  $GM$ is a generator matrix of $\C'$.

Now, we recall some definitions related to GRS and EGRS codes.
Let $\bm{a}=(a_{1},a_{2},\dots,a_{n})\in \F_q^n$
with $a_1,a_2,\dots,a_n$ distinct elements and $\bm{v}=(v_1,v_2,\dots,v_n)\in (\F_{q}^*)^n$, where $\F_{q}^*=\F_{q}\setminus \{0\}$.
For an integer $1\leq k\leq n$, the GRS code associated with $\bm{a}$ and $\bm{v}$ is defined as follows:
$$GRS_{k}(\bm{a},\bm{v})=\{(v_{1}f(a_{1}),v_{2}f(a_{2}),\dots,v_{n}f(a_{n})): f(x)\in \F_{q}[x],\ \deg(f(x))\leq k-1\}.$$
It is well known that $GRS_{k}(\bm{a},\bm{v})$ is an $[n,k]_{q}$ MDS code and the corresponding dual code is also an MDS code.
The generator matrix of $GRS_k(\bm{a},\bm{v})$ is
$$G_k(\bm{a},\bm{v})=(\bm{g}_0,\bm{g}_1,\dots,\bm{g}_{k-1})^\mathcal{T},$$ 
where $\bm{g}_i=(v_1a_1^i,v_2a_2^i,\dots,v_na_n^{i})$.
If $\bm{a}\in (\F_{q}^*)^n$,
let $GRS_{k,k_1}(\bm{a},\bm{v})$ be a code with generator matrix
$G_{k,k_1}(\bm{a},\bm{v})=(\bm{g}_{k_1},\bm{g}_{k_1+1},\dots,\bm{g}_{k_1+k-1})^\mathcal{T}$.
Put $\bm{v}'=(v_1',v_2',\dots,v_n')\in (\F_q^*)^n$, where $v_i'=v_ia_i^{k_1}$ for $1\leq i\leq n$, then $GRS_{k,k_1}(\bm{a},\bm{v})=GRS_{k}(\bm{a},\bm{v}')$. Then $GRS_{k,k_1}(\bm{a},\bm{v})$ is also an $[n,k]_q$ GRS code.

Moreover, for an integer $1\leq k\leq n$, the EGRS code associated with $\bm{a}$ and $\bm{v}$ is defined as follows:
$$GRS_{k}(\bm{a},\bm{v},\infty)=\{(v_{1}f(a_{1}),v_{2}f(a_{2}),\dots,v_{n}f(a_{n}),f_{k-1}):\ f(x)\in \F_{q}[x],\ \deg(f(x))\leq k-1\},$$
where $f_{k-1}$ is the coefficient of $x^{k-1}$ in $f(x)$.
It is well known that $GRS_{k}(\bm{a},\bm{v},\infty)$ is an $[n+1,k]_{q}$ MDS code and the corresponding dual code is also an MDS code.
The code $GRS_k(\bm{a},\bm{v},\infty)$ has a generator matrix
$$G_{k}(\bm{a},\bm{v},\infty)=[(\bm{g}_0,\bm{g}_1,\dots,\bm{g}_{k-1})^\mathcal{T}|\infty^T],$$
where $\bm{g}_i=(v_1a_1^i,v_2a_2^i,\dots,v_na_n^{i})$ and $\infty=(0,0,\dots,1)\in \F_{q}^k$.
If $\bm{a}\in (\F_{q}^*)^n$,
let $GRS_{k,k_1}(\bm{a},\bm{v},\infty)$ be a linear code with generator matrix
$G_{k,k_1}(\bm{a},\bm{v},\infty)=[(\bm{g}_{k_1},\bm{g}_{k_1+1},\dots,\bm{g}_{k_1+k-1})^\mathcal{T}|\infty^T]$.
Put $\bm{v}'=(v_1',v_2',\dots,v_n')\in (\F_q^*)^n$, where $v_i'=v_ia_i^{k_1}$ for $1\leq i\leq n$, then $GRS_{k,k_1}(\bm{a},\bm{v},\infty)=GRS_{k}(\bm{a},\bm{v}',\infty)$. Hence, $GRS_{k,k_1}(\bm{a},\bm{v},\infty)$ is also an $[n+1,k]_q$ EGRS code.
There are some important results in the literature on GRS codes and EGRS codes. We review them here.

\begin{lemma}\label{lem equiv}
(\cite[Corollary 3.3]{RefJ (2023) Liu Galois})
Let $\bm{a}=(a_{1},a_{2},\dots,a_{n})\in \F_q^n$
with $a_1,a_2,\dots,a_n$ distinct elements and $\bm{v}=(v_1,v_2,\dots,v_n)\in (\F_{q}^*)^n$.
Then for any $a\in \F_q^*$, $b\in \F_q$ and $\lambda\in \F_q^*$,
we have $$GRS_k(\bm{a},\bm{v})=GRS_k(a\bm{a}+b\bm{1},\lambda \bm{v})$$
and
$$GRS_k(\bm{a},\bm{v},\infty)=GRS_k(a\bm{a}+b\bm{1},a^{1-k}\bm{v},\infty),$$
where $a\bm{a}+b\bm{1}=(aa_1+b,aa_2+b,\dots,aa_n+b)\in \F_q^n$ and $\lambda\bm{v}=(\lambda v_1, \lambda v_2, \dots, \lambda v_n)\in (\F_q^*)^n$.
\end{lemma}

\begin{remark}\label{rem 1}
By Lemma \ref{lem equiv}, we have the following results.
\begin{itemize}
\item[(1)]
For a given GRS code $GRS_k(\bm{a},\bm{v})$ of length $n<q$, there exists $b\in \F_q$ such that $\bm{a}+b\bm{1}\in (\F_q^*)^n$, so when $n<q$, we can assume that $\bm{a}\in (\F_q^*)^n$ without loss of generality.
\item[(2)]
For a given EGRS code $GRS_k(\bm{a},\bm{v},\infty)$ of length $n+1<q+1$, there exists $b\in \F_q$ such that $\bm{a}+b\bm{1}\in (\F_q^*)^n$, so when $n<q$, we can assume that $\bm{a}\in (\F_q^*)^n$ without loss of generality.
\end{itemize}
\end{remark}

When $n\leq q$, we have that GRS codes are equivalent to EGRS codes.

\begin{lemma}\label{lem GRS=EGRS}
(\cite{RefJ (1996) Pellikaan} and \cite{RefJ (2022) equivalence})
If $n\leq q$, then a linear code with length $n$ is GRS if and only if it is EGRS.
\end{lemma}

In \cite{RefJ (2018) K good},
for a linear code $\C$ with generator matrix $G$,
Guenda et al. gave the relation between $dim(\Hull_H(\C))$ and the value of $rank(GG^\dagger)$.

\begin{lemma} (\cite[Proposition 3.2]{RefJ (2018) K good})\label{lem Hull=GG}
Let $\C$ be an $[n,k,d]_{q^2}$ linear code. If $H$ is a parity check matrix and $G$ is a generator matrix of $\C$. Then $rank(HH^\dagger)$ and $rank(GG^\dagger)$ are independent of $H$ and $G$ such that
$$rank(HH^\dagger)=n-k-dim(\Hull_H(\C))=n-k-dim(\Hull_H(\C^{\bot_H}))$$
and
$$rank(GG^\dagger)=k-dim(\Hull_H(\C))=k-dim(\Hull_H(\C^{\bot_H})).$$
\end{lemma}

In \cite{RefJ (2022) Luo} or \cite{RefJ (2023) Chen}, the authors
gave the following lemma using equivalent codes.

\begin{lemma}(\cite[Corollary 2.2]{RefJ (2023) Chen} or \cite[Theorem 6]{RefJ (2022) Luo})\label{lem 1}
Let $q>2$ be a prime power. If there exists an $[n,k]_{q^2}$ code $\C$ with $l$-dimension Hermitian hull, then there exists an equivalent $[n,k]_{q^2}$ code $\C'$ with $s$-dimensional Hermitian hull for each $0\leq s\leq l$.
\end{lemma}

\begin{remark}\label{rem 222}
Lemma \ref{lem 1} shows that for a linear code $\C$, if we find a lower bound $\delta$ on its Hermitian hull dimension, then there exists an equivalent code $\C'$ with $s$-dimensional Hermitian hull for each $0\leq s\leq \delta$.
\end{remark}

In \cite{RefJ (2024) Anderson},
Anderson et al. consider the more general relative hulls. For $e$-Galois hulls the following result is obtained.

\begin{lemma}\label{lem Galois}
(\cite[Proposition 3.8]{RefJ (2024) Anderson})
Let $q=p^m>2$ and $\C$ be an $[n,k]_q$ code
with $l$-dimensional $e$-Galois hull, where $1\leq l\leq k$ and $0\leq e\leq m-1$.
If there exists $x\in \F_q^*$ such that $x^{p^e+1}\neq 1$,
then
there exists an equivalent $[n,k]_{q}$ code $\C'$ with $(l-1)$-dimensional $e$-Galois hull.
\end{lemma}

By Lemma \ref{lem Galois}, we have the following corollary, which is a generalisation of Lemma \ref{lem 1}.

\begin{corollary}\label{cor Galois}
Let $q=p^m>4$ be a prime power. If there exists an $[n,k]_{q}$ code $\C$ with $l$-dimension $e$-Galois hull, where $0\leq e\leq m-1$, then there exists an equivalent $[n,k]_{q}$ code $\C'$ with $s$-dimensional $e$-Galois hull for each $0\leq s\leq l$.
\end{corollary}

\begin{proof}
Let $w$ be a primitive element of $\F_{q}$.
Then $w$ is a primitive $(q-1)$-th root of unity.
It follows that
$w^{p^e+1}=1$ if and only if
 $(q-1)\mid (p^e+1)$.
 Note that $(q-1)\nmid (p^e+1)$ can always be met for $q>4$ and $0\leq e\leq m-1$.
 By Lemma \ref{lem Galois}, we have the desired result.
 This completes the proof.
\end{proof}

\section{Three classes of GRS codes from known GRS codes}\label{sec3}

In this section, we present three classes of GRS codes that can be obtained from known GRS codes and determine the dimension of the hulls of the obtained GRS codes.
Now, we always prove all the results in the Hermitian case.
 Moreover, by Corollary \ref{cor Galois}, some results still hold for more general Galois inner products, which we note as corollaries.

\subsection{Extend the length of the GRS code}

In \cite{RefJ (2022) Luo}, for a given $[n,k,d]_{q^2}$ code with $l$-dimensional Hermitian hull, Luo et al. extended its length to obtain an $[n+1,k,d']_{q^2}$ code with $(l+1)$-dimensional Hermitian hull,
where $d\leq d'\leq d+1$.
Inspired by this result, we study the Hull-variable problem of extending the length of GRS codes.
Different from the results of Luo et al. the codes we obtain are always MDS codes.

\begin{proposition}\label{pro GRS 1}
Let $q>2$ be a prime power and $GRS_k(\bm{a},\bm{v})$ be an $[n,k]_{q^2}$ GRS code with generator matrix $G=G_{k}(\bm{a},\bm{v})$.
For any $\lambda\in \F_q^*$, define
$$S_\lambda=\begin{pmatrix}
O_{(k-1)\times (k-1)} & O_{(k-1)\times 1}\\
O_{1\times (k-1)} & \lambda\\
\end{pmatrix}\ and \
\bar{S}_\lambda=\begin{pmatrix}
\lambda & O_{1\times (k-1)}\\
O_{(k-1)\times 1} & O_{(k-1)\times(k-1)}\\
\end{pmatrix}.$$
Then the following statements hold.
\begin{itemize}
\item[(1)]
There exists an $[n+1,k]_{q^2}$ EGRS code with $s$-dimensional Hermitian hull for each $0\leq s\leq k-rank(GG^\dagger+S_\lambda)$.
\item[(2)]
If $n<q^2$,
then there exists an $[n+1,k]_{q^2}$ GRS code  with $s$-dimensional Hermitian hull for each $0\leq s\leq k-rank(GG^\dagger+\bar{S}_\lambda)$.
\item[(3)]
If $n<q^2$,
then there exists an $[n+2,k]_{q^2}$ EGRS code  with $s$-dimensional Hermitian hull for each $0\leq s\leq k-rank(GG^\dagger+S_{\lambda_1}+\bar{S}_{\lambda_2})$, where $\lambda_1,\lambda_2\in \F_{q}^*$.
\end{itemize}
\end{proposition}

\begin{proof}
(1)
For any $\lambda\in \F_q^*$, there exists $\xi$ such that $\xi^{q+1}=\lambda^{-1}$.
Let $GRS_k(\bm{a},\xi\bm{v},\infty)$ be an EGRS code with generator matrix
$G_k(\bm{a},\xi\bm{v},\infty)=[(\xi\bm{g}_0,\xi\bm{g}_1,\dots,\xi\bm{g}_{k-1})^\mathcal{T}|\infty^T]$,
where $\bm{g}_i$ and $\infty$ are defined as above.
Then we have
$$G_k(\bm{a},\xi\bm{v},\infty)G_k(\bm{a},\xi\bm{v},\infty)^\dagger=\lambda^{-1} GG^\dagger+\lambda^{-1}S_\lambda.$$
Hence $rank(G_k(\bm{a},\xi\bm{v},\infty)G_k(\bm{a},\xi\bm{v},\infty)^\dagger)=rank(GG^\dagger+S_\lambda).$
By Lemma \ref{lem Hull=GG}, $GRS_k(\bm{a},\xi\bm{v},\infty)$ is an $[n+1,k]_{q^2}$ EGRS code with $(k-rank(GG^\dagger+S_\lambda))$-dimensional Hermitian hull.
The desired result follows from Lemma \ref{lem 1}.

(2)
Since $n<q^2$, by Remark \ref{rem 1}, we assume $\bm{a}\in (\F_{q^2}^*)^n$ without loss of generality.
For any $\lambda\in \F_q^*$, there exists $\xi$ such that $\xi^{q+1}=\lambda$.
Let $\tilde{\bm{a}}=(0,\bm{a})\in \F_{q^2}^{n+1}$, $\tilde{\bm{v}}=(\xi,\bm{v})\in (\F_{q^2}^*)^{n+1}$.
Then $GRS_k(\tilde{\bm{a}},\tilde{\bm{v}})$ be a GRS code with generator matrix
$G_k(\tilde{\bm{a}},\tilde{\bm{v}})=[\bm{u}^T|(\bm{g}_0,\bm{g}_1,\dots,\bm{g}_{k-1})^\mathcal{T}]$,
where $\bm{u}=(\xi,0\dots,0)\in \F_{q^2}^k$.
Then we have
$$G_k(\tilde{\bm{a}},\tilde{\bm{v}})G_k(\tilde{\bm{a}},\tilde{\bm{v}})^\dagger= GG^\dagger+\bar{S}_\lambda.$$
Hence $rank(G_k(\tilde{\bm{a}},\tilde{\bm{v}})G_k(\tilde{\bm{a}},\tilde{\bm{v}})^\dagger)=rank(GG^\dagger+\bar{S}_\lambda).$
By Lemma \ref{lem Hull=GG}, $GRS_k(\tilde{\bm{a}},\tilde{\bm{v}})$ is an $[n+1,k]_{q^2}$ GRS code with $(k-rank(GG^\dagger+\bar{S}_\lambda))$-dimensional Hermitian hull.
 By Lemma \ref{lem 1}, we have the desired result.

(3) Combining (1) and (2) gives the desired result. This completes the proof.
\end{proof}

Note that $rank(GG^\dagger+S_\lambda)\leq rank(GG^\dagger)+rank(S_\lambda)=rank(GG^\dagger)+1$ and
$rank(GG^\dagger+S_{\lambda_1}+\bar{S}_{\lambda_2})\leq rank(GG^\dagger)+rank(S_{\lambda_1}+\bar{S}_{\lambda_2})=rank(GG^\dagger)+2$,
 we have the following corollary.

\begin{corollary}\label{cor GRS 1}
Let $q>2$ be a prime power.
If there exists an $[n,k]_{q^2}$ GRS code with $l$-dimensional Hermitian hull, where $0\leq l\leq k$, then for any $1\leq i\leq \min\{l,q^2+1-n\}$, there exists an $[n+i,k]_{q}$ MDS code with $s$-dimensional Hermitian hull for each $0\leq s\leq l-i$.
\end{corollary}

By Corollary \ref{cor Galois},
the following result for Galois hull can be proved similarly.

\begin{corollary}
Let $q>4$ be a prime power.
If there exists an $[n,k]_{q}$ GRS code with $l$-dimensional $e$-Galois hull, where $0\leq l\leq k$ and $0\leq e\leq m-1$, then for any $1\leq i\leq \min\{l,q+1-n\}$, there exists an $[n+i,k]_{q}$ MDS code with $s$-dimensional $e$-Galois hull for each $0\leq s\leq l-i$.
\end{corollary}

Under certain conditions, by Proposition \ref{pro GRS 1}, we obtain some Hermitian self-orthogonal MDS codes.

\begin{theorem}\label{th GRS 1}
Let $GRS_k(\bm{a},\bm{v})$ be an $[n,k]_{q^2}$ GRS code with $(k-1)$-dimensional Hermitian hull.
Then the following statements hold.
\begin{itemize}
\item[(1)]
If $\Hull_H(GRS_k(\bm{a},\bm{v}))=GRS_{k-1}(\bm{a},\bm{v})$,
then there exists an $[n+1,k]_{q^2}$ Hermitian self-orthogonal EGRS code.
\item[(2)]
If $n<q^2$ and $\Hull_H(GRS_k(\bm{a},\bm{v}))=GRS_{k-1,1}(\bm{a},\bm{v})$,
then there exists an $[n+1,k]_{q^2}$ Hermitian self-orthogonal GRS code.
\end{itemize}
\end{theorem}

\begin{proof}
By Proposition \ref{pro GRS 1}, the proof of (2) is exactly similar to (1), then we only prove (1).
Let the GRS code $GRS_k(\bm{a},\bm{v})$ have a generator matrix $G_{k}(\bm{a},\bm{v})=(\bm{g}_0,\bm{g}_1,\dots,\bm{g}_{k-1})^\mathcal{T}$,
where $\bm{g}_i$ is defined as above.
Since $\Hull_H(GRS_k(\bm{a},\bm{v}))=GRS_{k-1}(\bm{a},\bm{v})$, it follows that $G_{k-1}(\bm{a},\bm{v})=(\bm{g}_0,\bm{g}_1,\dots,\bm{g}_{k-2})^\mathcal{T}$ is a basis of the code $\Hull_H(GRS_k(\bm{a},\bm{v}))$.
Note that $\bm{g}_i\bm{g}_j^\dagger=0$ if and only if $\bm{g}_j\bm{g}_i^\dagger=0$,
then we have $\bm{g}_i\bm{g}_j^\dagger=0$ for $0\leq i\leq j\leq k-1$ except $i=j=k-1$ and $\bm{g}_{k-1}\bm{g}_{k-1}^\dagger\neq 0$.
Since $-(\bm{g}_{k-1}\bm{g}_{k-1}^\dagger)^{-1}\in \F_{q}^*$, then there exists
$\xi\in \F_{q^2}^*$ such that $\xi^{q+1}=-(\bm{g}_{k-1}\bm{g}_{k-1}^\dagger)^{-1}$.
By Proposition \ref{pro GRS 1}, it is easy to check that
$GRS_k(\bm{a},\xi\bm{v},\infty)$ is an $[n+1,k]_{q^2}$ Hermitian self-orthogonal EGRS code.
This completes the proof.
\end{proof}

\begin{remark}
In \cite[Theorems 9, 10, 11, 12 and 13]{RefJ (2024) Luo MDS}, Luo et al. constructed some GRS codes satisfying $\Hull_H(GRS_k(\bm{a},\bm{v}))=GRS_{k-1}(\bm{a},\bm{v})$. From Theorem \ref{th GRS 1}, we can directly obtain many Hermitian self-orthogonal MDS codes.
\end{remark}

\subsection{Increase the dimension of the GRS code}

In this subsection, we study the Hull-variable problem of
  increasing the dimension of GRS codes.
  Moreover, from the currently known Hermitian self-orthogonal GRS codes,
  we can directly obtain MDS codes with larger dimensions and which have
  determined Hermitian hull dimensions.

\begin{proposition}\label{pro 3}
Let $q>2$ be a prime power and $GRS_k(\bm{a},\bm{v})$ be an $[n,k]_{q^2}$ GRS code
 with generator matrix $G=G_{k}(\bm{a},\bm{v})$
and $l$-dimensional Hermitian hull,
where $0\leq l\leq k$.
Let $\bm{g}_i$ be defined as above, then the following statements hold.
\begin{itemize}
\item[(1)]
If $\bm{g}_k\in GRS_k(\bm{a},\bm{v})^{\bot_H}$ and $\bm{g}_k\bm{g}_k^\dagger=0$ $(resp.\ \bm{g}_{-1}\in GRS_k(\bm{a},\bm{v})^{\bot_H}\ and\ \bm{g}_{-1}\bm{g}_{-1}^\dagger=0\ with\ n<q^2)$,
 then
there exists an $[n,k+1]_{q^2}$ GRS code with $s$-dimensional Hermitian hull for each $0\leq s\leq l+1$.
\item[(2)]
If $\bm{g}_k\notin (\Hull_H(GRS_k(\bm{a},\bm{v})))^{\bot_H}$ $(resp.\ \bm{g}_{-1}\notin (\Hull_H(GRS_k(\bm{a},\bm{v})))^{\bot_H}\ with\ n<q^2)$,
then there exists an $[n,k+1]_{q^2}$ GRS code  with $s$-dimensional Hermitian hull for each $0\leq s\leq l-1$.
\item[(3)]
If $\bm{g}_k\in (\Hull_H(GRS_k(\bm{a},\bm{v})))^{\bot_H}$ $(resp.\ \bm{g}_{-1}\in (\Hull_H(GRS_k(\bm{a},\bm{v})))^{\bot_H}\ with\ n<q^2)$,
then there exists an $[n,k+1]_{q^2}$ GRS code  with $s$-dimensional Hermitian hull for each  $0\leq s\leq l$.
\item[(4)]
If $\bm{g}_{k}\bm{g}_k^{\dagger}\neq 0$ $(resp.\ \bm{g}_{-1}\bm{g}_{-1}^{\dagger}\neq 0\ with\ n<q^2)$,
then there exists an $[n,k+1]_{q^2}$ GRS code with $s$-dimensional Hermitian hull for each $0\leq s\leq k-rank(GSG^\dagger)$, where $S=\bm{g}_k\bm{g}_k^\dagger E_n-\bm{g}_k^\dagger\bm{g}_k$ $(resp.\ S=\bm{g}_{-1}\bm{g}_{-1}^\dagger E_n-\bm{g}_{-1}^\dagger\bm{g}_{-1})$ be an $n\times n$ matrix with rank $n-1$.
\end{itemize}
\end{proposition}

\begin{proof}
When $n<q^2$, by Remark \ref{rem 1}, we can assume that $\bm{a}\in (\F_{q^2}^*)^n$.
Note that we can increase the dimension of this code by adding $\bm{g}_k$ or $\bm{g}_{-1}\ (n<q^2)$ to the base of $GRS_k(\bm{a},\bm{v})$.
Since the proofs are similar, here we will only prove the case of adding $\bm{g}_k$.
Suppose that $\{\bm{g}_{h_1},\bm{g}_{h_2},\dots,\bm{g}_{h_l}\}$ be a basis of $\Hull_H(GRS_k(\bm{a},\bm{v}))$,
where $0\leq h_i\leq k-1$ for $1\leq i\leq l$.
Let $G_1=(\bm{g}_{h_1},\bm{g}_{h_2},\dots,\bm{g}_{h_l})^\mathcal{T}$.

(1)
If $\bm{g}_k\in GRS_k(\bm{a},\bm{v})^{\bot_H}$, then $G\bm{g}_k^\dagger=O_{k\times 1}$.
Since $G\bm{g}_k^\dagger=O_{k\times 1}$ and $\bm{g}_k\bm{g}_k^\dagger=0$, it follows that
$$G_{k+1}(\bm{a},\bm{v})G_{k+1}(\bm{a},\bm{v})^\dagger=
\begin{pmatrix}
GG^\dagger & G\bm{g}_k^\dagger\\
\bm{g}_kG^\dagger & \bm{g}_k\bm{g}_k^\dagger\\
\end{pmatrix}
=\begin{pmatrix}
GG^\dagger  & O_{k\times 1}\\
O_{1\times k} & 0\\
\end{pmatrix}.$$
Note that $rank(GG^\dagger)=k-l$,
then by Lemma \ref{lem Hull=GG}, $GRS_{k+1}(\bm{a},\bm{v})$ is an $[n,k+1]_{q^2}$ GRS code with $(l+1)$-dimensional Hermitian hull.
The desired result follows from Lemma \ref{lem 1}.

(2)
If $\bm{g}_k\notin (\Hull_H(GRS_k(\bm{a},\bm{v})))^{\bot_H}$, then $G_1\bm{g}_k^{\dagger}\neq O_{l\times 1}$,
it follows that there exists $1\leq t\leq l$ such that $\bm{g}_k\bm{g}_{h_t}^{\dagger}\neq 0$.
Let $\tilde{\bm{u}}=(0,0,\dots,0, \bm{g}_{h_t}\bm{g}_k^{\dagger})\in \F_{q^2}^{k+1}$ and
$\bm{u}=(\underbrace{0,\dots,0}_{h_t-1\ times },\bm{g}_{h_t}\bm{g}_k^\dagger, \underbrace{0,\dots,0}_{k-h_t\ times})\in\F_{q^2}^k$.
By (1), we can find that the $h_t$-th row and $h_t$-th column of the matrix $G_{k+1}(\bm{a},\bm{v})G_{k+1}(\bm{a},\bm{v})^\dagger$ are $\tilde{\bm{u}}$ and $\tilde{\bm{u}}^\dagger$, respectively. Then we can perform the elementary transformation such that
$$rank(
\begin{pmatrix}
GG^\dagger & G\bm{g}_k^\dagger\\
\bm{g}_kG^\dagger & \bm{g}_k\bm{g}_k^\dagger\\
\end{pmatrix})=
rank(
\begin{pmatrix}
GG^\dagger & \bm{u}^T\\
\bm{u}^q & 0\\
\end{pmatrix})
=rank(GG^\dagger)+2=k-l+2.
$$
 By Lemma \ref{lem Hull=GG}, $GRS_{k+1}(\bm{a},\bm{v})$ is an $[n,k+1]_{q^2}$ GRS code with $(l-1)$-dimensional Hermitian hull.
 The desired result follows from Lemma \ref{lem 1}.

(3)
Let $\{\bm{g}_{z_1},\bm{g}_{z_2},\dots,\bm{g}_{z_{k-l}}\}=\{\bm{g}_0,\bm{g}_1,\dots,\bm{g}_{k-1}\}\backslash \{\bm{g}_{h_1},\bm{g}_{h_2},\dots,\bm{g}_{h_l}\}$ and $G_2=(\bm{g}_{z_1},\bm{g}_{z_2},\dots,\bm{g}_{z_{k-l}})^\mathcal{T}$.
It is easy to check that $G_2G_2^\dagger$ is an invertible $(k-l)\times (k-l)$ matrix.
If $\bm{g}_k\in (\Hull_H(GRS_k(\bm{a},\bm{v})))^{\bot_H}$, then $G_1\bm{g}_k^{\dagger}=O_{l\times 1}$,
it follows that $\bm{g}_k\bm{g}_{h_i}^\dagger=0$ for $1\leq i\leq l$.
 Then we can perform the elementary transformation such that
 \[\begin{split}
rank(
\begin{pmatrix}
GG^\dagger & G\bm{g}_k^\dagger\\
\bm{g}_kG^\dagger & \bm{g}_k\bm{g}_k^\dagger\\
\end{pmatrix})
&=rank(
\begin{pmatrix}
G_2G_2^\dagger & G_2\bm{g}_k^\dagger\\
\bm{g}_kG_2^\dagger & \bm{g}_k\bm{g}_k^\dagger\\
\end{pmatrix}
)\\
&=rank(
\begin{pmatrix}
E_{k-l} & O_{(k-l)\times 1}\\
-\bm{g}_kG_2^\dagger (G_2G_2^\dagger)^{-1} & 1\\
\end{pmatrix}
\begin{pmatrix}
G_2G_2^\dagger & G_2\bm{g}_k^\dagger\\
\bm{g}_kG_2^\dagger & \bm{g}_k\bm{g}_k^\dagger\\
\end{pmatrix}
)\\
&=rank(
\begin{pmatrix}
G_2G_2^\dagger & G_2\bm{g}_k^\dagger\\
O_{1\times (k-l)} & \bm{g}_k\bm{g}_k^\dagger-\bm{g}_kG_2^\dagger (G_2G_2^\dagger)^{-1} G_2\bm{g}_k^\dagger\\
\end{pmatrix})
\\
&\leq rank(G_2G_2^\dagger)+1=k-l+1.\\
	\end{split}\]
By lemmas \ref{lem Hull=GG} and \ref{lem 1}, we have the desired result.

(4)
Note that
\[\begin{split}
rank(
\begin{pmatrix}
GG^\dagger & G\bm{g}_k^\dagger\\
\bm{g}_kG^\dagger & \bm{g}_k\bm{g}_k^\dagger\\
\end{pmatrix})
&=rank(
\begin{pmatrix}
E_{k-1} & -G\bm{g}_k^\dagger(\bm{g}_k\bm{g}_k^\dagger)^{-1}\\
O_{(k-1)\times 1} & 1\\
\end{pmatrix}
\begin{pmatrix}
GG^\dagger & G\bm{g}_k^\dagger\\
\bm{g}_kG^\dagger & \bm{g}_k\bm{g}_k^\dagger\\
\end{pmatrix}
)\\
&=rank(
\begin{pmatrix}
GG^\dagger-G\bm{g}_k^\dagger(\bm{g}_k\bm{g}_k^\dagger)^{-1}\bm{g}_kG^\dagger
 & O_{1\times (k-1)}\\
\bm{g}_kG^\dagger & \bm{g}_k\bm{g}_k^\dagger\\
\end{pmatrix})\\
&=rank(G(\bm{g}_k\bm{g}_k^\dagger E_n-\bm{g}_k^\dagger\bm{g}_k)G^\dagger)+1,\\
	\end{split}\]
then by Lemma \ref{lem Hull=GG}, $GRS_{k+1}(\bm{a},\bm{v})$ is an $[n,k+1]_{q^2}$ GRS code with $(k-rank(G(\bm{g}_k\bm{g}_k^\dagger E_n-\bm{g}_k^\dagger\bm{g}_k)G^\dagger))$-dimensional Hermitian hull.
Note that
\[\begin{split}
rank(
\begin{pmatrix}
\bm{g}_k\bm{g}_k^\dagger E_n & \bm{g}_k^\dagger\\
\bm{g}_k &1\\
\end{pmatrix})
&=rank(
\begin{pmatrix}
E_{n} & -\bm{g}_k^\dagger\\
O_{n\times 1} & 1\\
\end{pmatrix}
\begin{pmatrix}
\bm{g}_k\bm{g}_k^\dagger E_n & \bm{g}_k^\dagger\\
\bm{g}_k &1\\
\end{pmatrix}
)\\
&=rank(
\begin{pmatrix}
\bm{g}_k\bm{g}_k^\dagger E_n-\bm{g}_k^\dagger\bm{g}_k & O_{1\times n}\\
\bm{g}_k &1\\
\end{pmatrix})\\
&=rank(\bm{g}_k\bm{g}_k^\dagger E_n-\bm{g}_k^\dagger\bm{g}_k)+1\\
	\end{split}\]
and
\[\begin{split}
rank(
\begin{pmatrix}
\bm{g}_k\bm{g}_k^\dagger E_n & \bm{g}_k^\dagger\\
\bm{g}_k &1\\
\end{pmatrix})
&=rank(
\begin{pmatrix}
E_{n} & O_{1\times n}\\
-\bm{g}_k(\bm{g}_k\bm{g}_k^\dagger)^{-1}E_n & 1\\
\end{pmatrix}
\begin{pmatrix}
\bm{g}_k\bm{g}_k^\dagger E_n & \bm{g}_k^\dagger\\
\bm{g}_k &1\\
\end{pmatrix}
)\\
&=rank(
\begin{pmatrix}
\bm{g}_k\bm{g}_k^\dagger E_n & \bm{g}_k^\dagger\\
O_{n\times 1} &0\\
\end{pmatrix})\\
&=n,\\
\end{split}\]
we have that $rank(\bm{g}_k\bm{g}_k^\dagger E_n-\bm{g}_k^\dagger\bm{g}_k)=n-1$.
The desired result follows from Lemma \ref{lem 1}. This completes the proof.
\end{proof}

By Propositions \ref{pro 3} (2) and \ref{pro 3} (3), we have the following corollary.

\begin{corollary}\label{cor GRS 2}
Let $q>2$ be a prime power.
If there exists an $[n,k]_{q^2}$ GRS code with $l$-dimensional Hermitian hull,
where $0\leq l\leq k$,
then for any $1\leq i\leq \min\{l,n-k\}$, there exists an $[n,k+i]_{q^2}$ MDS code with $s$-dimensional Hermitian hull for each $0\leq s\leq l-i$.
\end{corollary}

Note that $$rank(
\begin{pmatrix}
GG^\ddagger & G\bm{g}_k^\ddagger\\
\bm{g}_kG^\ddagger & \bm{g}_k\bm{g}_k^\ddagger\\
\end{pmatrix})=
rank(
\begin{pmatrix}
GG^\ddagger & O_{k\times 1}\\
O_{1\times k} & 0\\
\end{pmatrix}
+
\begin{pmatrix}
O_{k\times k} & G\bm{g}_k^\ddagger\\
\bm{g}_kG^\ddagger & \bm{g}_k\bm{g}_k^\ddagger\\
\end{pmatrix})
\leq rank(GG^\ddagger)+2,
$$
then by Corollary \ref{cor Galois},  the following result for Galois hull can be proved similarly.

\begin{corollary}
Let $q>4$ be a prime power.
If there exists an $[n,k]_{q}$ GRS code with $l$-dimensional $e$-Galois hull, where $0\leq l\leq k$ and $0\leq e\leq m-1$, then for any $1\leq i\leq \min\{l,n-k\}$, there exists an $[n,k+i]_{q}$ MDS code with $s$-dimensional $e$-Galois hull for each $0\leq s\leq l-i$.
\end{corollary}

Now we generalize to consider the case of $k=q^2$. We know that for any $1\leq k\leq q^2$, $G_k(\bm{a},\bm{v})G_k(\bm{a},\bm{v})^\dagger$ is a $k$-order leading principal submatrix of $G_{q^2}(\bm{a},\bm{v})G_{q^2}(\bm{a},\bm{v})^\dagger$.

\begin{lemma}\label{lem q^2}
For any $1\leq k\leq q^2$,
let $\bm{a}$, $\bm{v}$, $\bm{g}_i$ and $G_k(\bm{a},\bm{v})$ be defined as above.
Then
$$G_{q^2}(\bm{a},\bm{v})G_{q^2}(\bm{a},\bm{v})^\dagger=
\begin{pmatrix}
A_{0,0}& A_{1,0}& \dots & A_{q-1,0}\\
A_{0,1}& A_{1,1}& \dots & A_{q-1,1}\\
\vdots & \vdots & \ddots &\vdots\\
A_{0,q-1}& A_{1,q-1}& \dots & A_{q-1,q-1}\\
\end{pmatrix},$$
where
$A_{i,j}=
\begin{pmatrix}
\bm{g}_i\bm{g}_j^\dagger& \bm{g}_i\bm{g}_{j+1}^\dagger& \dots & \bm{g}_i\bm{g}_{j+q-1}^\dagger\\
\bm{g}_{i+1}\bm{g}_j^\dagger& \bm{g}_{i+1}\bm{g}_{j+1}^\dagger& \dots & \bm{g}_{i+1}\bm{g}_{j+q-1}^\dagger\\
\vdots & \vdots & \ddots &\vdots\\
\bm{g}_{i+q-1}\bm{g}_j^\dagger& \bm{g}_{i+q-1}\bm{g}_{j+1}^\dagger& \dots & \bm{g}_{i+q-1}\bm{g}_{i+q-1}^\dagger\\
\end{pmatrix}$
be an $q\times q$ matrix over $\F_{q^2}$ for $0\leq i,j\leq q-1$.
Moreover, for $0\leq s,t\leq q-1$ except $s=t=0$, $\bm{g}_s\bm{g}_t^\dagger$ occurs twice in $A_{s,t}$ and only once in $A_{i,j}$, where $0\leq i,j\leq q-1$ except $i=s$ and $j=t$.
\end{lemma}

\begin{proof}
For any $1\leq i,j\leq q^2$, we have that the element of the $i$-th row and $j$-th column of the matrix $G_{q^2}(\bm{a},\bm{v})G_{q^2}(\bm{a},\bm{v})^\dagger$ is $\bm{g}_{i-1}\bm{g}_{j-1}^\dagger$.
Note that
$$\bm{g}_i\bm{g}_j^\dagger=\sum_{l-1}^{n}v_l^{q+1}a_l^{qi+j}=\sum_{l-1}^{n}v_l^{q+1}a_l^{q(i-q)+j+1}=\bm{g}_{i-q}\bm{g}_{j+1}^\dagger$$
and
$$\bm{g}_{i}\bm{g}_{j}^\dagger=\sum_{l-1}^{n}v_l^{q+1}a_l^{qi+j}=\sum_{l-1}^{n}v_l^{q+1}a_l^{q(i+1)+j-q}=\bm{g}_{i+1}\bm{g}_{j-q}^\dagger,$$
then the element in any place of matrix $G_{q^2}(\bm{a},\bm{v})G_{q^2}(\bm{a},\bm{v})^\dagger$ is one of $\bm{g}_i\bm{g}_j^\dagger$ for $0\leq i,j\leq q-1$.
 Then the matrix $G_{q^2}(\bm{a},\bm{v})G_{q^2}(\bm{a},\bm{v})^\dagger$ is verified to be of the above form.  This completes the proof.
\end{proof}

When given a class of Hermitian self-orthogonal GRS codes, we have the following result.

\begin{theorem}\label{th 2q}
Let $q>2$ be a prime power.
Let $GRS_{k'}(\bm{a},\bm{v})$ be an $[n,k']_{q^2}$ Hermitian self-orthogonal GRS code,
where $n> q+1$ and $k'\leq q-1$.
Then the following statements hold.
\begin{itemize}
\item[(1)]
For $k'\leq k\leq q$,
there exists an $[n,k]_{q^2}$ GRS code with $s$-dimensional Hermitian hull for each $0\leq s\leq 2k'-k$,
where $2k'-k\geq 0$.
\item[(2)]
For $q< k\leq \min\{q+k'-1,n\}$,
there exists an $[n,k]_{q^2}$ GRS code  with $s$-dimensional Hermitian hull for each $0\leq s\leq \max\{2k'-k,k+4k'-4q\}$,
where $\max\{2k'-k,k+4k'-4q\}\geq 0$.
\end{itemize}
\end{theorem}
\begin{proof}
(1)
By Corollary \ref{cor GRS 2}, we have the desired result.

(2) For $q< k\leq q+k'-1$, we have that $G_k(\bm{a},\bm{v})G_k(\bm{a},\bm{v})^\dagger$ is a $k$-order leading principal submatrix of $G_{q^2}(\bm{a},\bm{v})G_{q^2}(\bm{a},\bm{v})^\dagger$.
By Lemma \ref{lem q^2}, it is easy to check that $rank(G_k(\bm{a},\bm{v})G_k(\bm{a},\bm{v})^\dagger)\leq 4q-4k'$.
The desired result follows from Lemmas \ref{lem Hull=GG} and \ref{lem 1}. This completes the proof.
\end{proof}

When the dimension of the Hermitian self-orthogonal GRS code is $q-1$, by Theorem \ref{th 2q}, we have the following corollary.

\begin{corollary}\label{cor q-1}
Let $q>2$ be a prime power.
Let $GRS_{q-1}(\bm{a},\bm{v})$ be an $[n,q-1]_{q^2}$ Hermitian self-orthogonal GRS code,  where $n\geq \begin{cases}
\frac{q^2-1}{2}, & if\ q\ odd;\\
\frac{q^2}{2}, & if\  q\ even.
\end{cases}$
Then the following statements hold.
\begin{itemize}
\item[(1)]
For $k=q$,
there exists an $[n,k]_{q^2}$ GRS code  with $s$-dimensional Hermitian hull for each $0\leq s\leq q-2$.
\item[(2)]
For $q+1\leq  k\leq 2q-2$,
there exists an $[n,k]_{q^2}$ GRS code  with $s$-dimensional Hermitian hull for each $0\leq s\leq k-4$.
\end{itemize}
\end{corollary}

\subsection{Extend the length and increase the dimension at the same time of the GRS code}

In \cite{RefJ (2022) Luo}, for a given $[n,k]_{q^2}$-code with $l$-dimensional Hermitian hull, Luo et al. obtain an $[n+1,k+1]_{q^2}$-code with $(l+1)$-dimensional Hermitian hull.
Inspired by this result, we study the Hull-variable problem of extending the length and increasing the dimension at the same time of GRS codes.

\begin{proposition}\label{pro 4}
Let $q>2$ be a prime power and $GRS_k(\bm{a},\bm{v})$ be an $[n,k]_{q^2}$ GRS code
 with generator matrix $G=G_{k}(\bm{a},\bm{v})$
and $l$-dimensional Hermitian hull,
where $0\leq l\leq k$.
Let $\bm{g}_i$ be defined as above, then the following statements hold.

\begin{itemize}
\item[(1)]
If $\bm{g}_k\in GRS_k(\bm{a},\bm{v})^{\bot_H}$ and $\bm{g}_k\bm{g}_k^\dagger\neq 0$ $(resp.\ \bm{g}_{-1}\in GRS_k(\bm{a},\bm{v})^{\bot_H}\ and\ \bm{g}_{-1}\bm{g}_{-1}^\dagger\neq 0\ with\ n<q^2)$,
 then
there exists an $[n+1,k+1]_{q^2}$ MDS code with $s$-dimensional Hermitian hull for each $0\leq s\leq l+1$.
\item[(2)]
If $\bm{g}_k\notin (\Hull_H(GRS_k(\bm{a},\bm{v})))^{\bot_H}$ $(resp.\ \bm{g}_{-1}\notin (\Hull_H(GRS_k(\bm{a},\bm{v})))^{\bot_H}\ with\ n<q^2)$,
then there exists an $[n+1,k+1]_{q^2}$ MDS code  with $s$-dimensional Hermitian hull for each $0\leq s\leq l-1$.
\item[(3)]
If $\bm{g}_k\in (\Hull_H(GRS_k(\bm{a},\bm{v})))^{\bot_H}$ $(resp.\ \bm{g}_{-1}\in (\Hull_H(GRS_k(\bm{a},\bm{v})))^{\bot_H}\ with\ n<q^2)$,
then there exists an $[n+1,k+1]_{q^2}$ MDS code  with $s$-dimensional Hermitian hull for each $0\leq s\leq l$.
\item[(4)]
There exists an $[n+1,k+1]_{q^2}$ MDS code  with $s$-dimensional Hermitian hull for each $0\leq s\leq k-rank(GSG^\dagger)$, where $S=(\lambda+\bm{g}_k\bm{g}_k^\dagger)E_n-\bm{g}_k^\dagger\bm{g}_k$ for $\lambda\in \F_q^*$ such that $\lambda+\bm{g}_k\bm{g}_k^\dagger\neq 0$
 $(resp.\ S=(\lambda+\bm{g}_{-1}\bm{g}_{-1}^\dagger)E_n-\bm{g}_{-1}^\dagger\bm{g}_{-1}$ for $\lambda\in \F_q^*$ such that $\lambda+\bm{g}_{-1}\bm{g}_{-1}^\dagger\neq 0\ with\ n<q^2)$.
\end{itemize}
\end{proposition}

\begin{proof}
Similarly, we only prove the case where $\bm{g}_k$ is added and extended to EGRS code.

(1) For any $\lambda\in \F_q^*$, there exists $\xi$ such that $\xi^{q+1}=\lambda^{-1}$.
By Propositions \ref{pro GRS 1} and \ref{pro 3},
we have that
$$
rank(G_{k+1}(\bm{a},\xi\bm{v},\infty)G_{k+1}(\bm{a},\xi\bm{v},\infty)^\dagger)
=rank(
\begin{pmatrix}
GG^\dagger & G\bm{g}_k^\dagger\\
\bm{g}_kG^\dagger & \bm{g}_k\bm{g}_k^\dagger+\lambda\\
\end{pmatrix}).
$$
If $\lambda=-\bm{g}_k\bm{g}_k^\dagger$, then $GRS_{k+1}(\bm{a},\xi\bm{v},\infty)$ is an $[n+1,k+1]_{q^2}$ EGRS code with $(l+1)$-dimensional Hermitian hull.
The desired result follows from Lemma \ref{lem 1}.

The proof of (2)(3)(4) is similar to Proposition \ref{pro 3}. We omit the details.
This completes the proof.
\end{proof}

\begin{remark}
Comparison of Proposition \ref{pro 4} and \cite[Proposition 9]{RefJ (2022) Luo}:
\begin{itemize}
\item
Different from \cite[ Proposition 9]{RefJ (2022) Luo}, we choose the added vector to be $\bm{g}_k$ $(resp.\ \bm{g}_{-1}\ with\ n\leq q)$, so the resulting code remains an MDS code.
\item
In \cite{RefJ (2022) Luo}, Luo et al. consider added vector $\bm{c}$ satisfying $\bm{c}\in \C^{\bot_H} \backslash \Hull_H(\C)$ such that $\bm{c}\bm{c}^{\dagger}\neq 0$, while Proposition \ref{pro 4} considers added vector $\bm{g}_k$ $(resp.\ \bm{g}_{-1}\ with\ n\leq q)$, which does not always need to satisfy the above condition. Thus, Proposition \ref{pro 4} differs from \cite[Proposition 9]{RefJ (2022) Luo}.
\end{itemize}
\end{remark}

By the proof of Propositions \ref{pro 3} (3) and \ref{pro 4} (3), we have the following corollary.

\begin{corollary}
Let $GRS_k(\bm{a},\bm{v})$ be an $[n,k]_{q^2}$ GRS code with $l$-dimensional Hermitian hull,
where $0\leq l\leq k$.
If $\bm{g}_k\in (\Hull_H(GRS_k(\bm{a},\bm{v})))^{\bot_H}$ $(resp.\ \bm{g}_{-1}\in (\Hull_H(GRS_k(\bm{a},\bm{v})))^{\bot_H}\ with\ n<q^2)$, then there exists an $[n\ or\ n+1,k+1]_{q^2}$ MDS code with $(l+1)$-dimensional Hermitian hull.
\end{corollary}

By Propositions \ref{pro 4} (2) and \ref{pro 4} (3), we have the following corollary.

\begin{corollary}\label{cor GRS 3}
Let $q>2$ be a prime power.
If there exists an $[n,k]_{q^2}$ GRS code with $l$-dimensional Hermitian hull,
where $0\leq l\leq k$,
then for any $1\leq i\leq \min\{l,q^2+1-n\}$, there exists an $[n+i,k+i]_{q^2}$ MDS code with $s$-dimensional Hermitian hull for each $0\leq s\leq l-i$.
\end{corollary}

 By Corollary \ref{cor Galois},
the following result for Galois hull can be proved similarly.

\begin{corollary}\label{cor GALOIS}
Let $q>4$ be a prime power.
If there exists an $[n,k]_{q}$ GRS code with $l$-dimensional $e$-Galois hull, where $0\leq l\leq k$ and $0\leq e\leq m-1$, then for any $1\leq i\leq \min\{l,q+1-n\}$, there exists an $[n+i,k+i]_{q}$ MDS code with $s$-dimensional $e$-Galois hull for each $0\leq s\leq l-i$.
\end{corollary}

\section{MDS codes with Hermitian hulls of arbitrary dimensions}\label{sec MDS}

In \cite{RefJ (2021) Ball determine}, Ball et al. transformed the existence of a Hermitian self-orthogonal GRS code into the existence of a polynomial with a given number of distinct zeros and
 proved Conjecture 11 in \cite{RefJ (2015) n=q^2+1(2)}.
 From this we know the minimum distance of
of the puncture code $P(\C)$ of the $[q^2+1,k]_{q^2}$ RS code,
where the concept of puncture codes was first introduced by Rains (see \cite{RefJ (1999) introduct 2}).

 $\bullet$ $\textbf{Conjecture 11 in \cite{RefJ (2015) n=q^2+1(2)}:}$
The minimum distance of the puncture code $P(\C)$ of the $[q^2+1,k]_{q^2}$ RS code $\C$
is
$$\begin{cases}
2k, & if\ 1\leq k\leq q/2;\\
(q+1)(k-(q-1)/2), & if\ (q+1)/2\leq k\leq q-1,\ q\ odd;\\
q(k+1-q/2), & if\ q/2\leq k\leq q-1,\ q\ even;\\
q^2+1, & if\ k=q,\\
\end{cases}$$
where $P(\C)=\{\bm{\lambda}=(\lambda_1,\lambda_2,\dots,\lambda_n)\in \F_q^n:\ \langle  \bm{x}\star \bm{y}^q,\bm{\lambda} \rangle_0=0,\ for\ all\ \bm{x},\bm{y}\in \C \}$.

 By \cite[Conjecture 11]{RefJ (2015) n=q^2+1(2)} and \cite[Theorem 2]{RefJ (2021) Ball determine},
 we can deduce that the minimum length of a $q^2$-ary Hermitian self-orthogonal GRS code of dimension $q-1$ is $\begin{cases}
\frac{q^2-1}{2}, & if\ q\ odd;\\
\frac{q^2}{2}, & if\  q\ even.
\end{cases}$
However, the construction of $q^2$-ary Hermitian self-orthogonal GRS codes of dimension $q-1$ remains an open problem.
Until now, only some special length $q^2$-ary Hermitian self-orthogonal GRS codes of dimension $q-1$ have been constructed.
In this section,
based on some known $q^2$-ary Hermitian self-orthogonal GRS codes with dimension $q-1$,
 we obtain some classes of MDS codes with Hermitian hulls of arbitrary dimensions.
 It is worth noting that the dimension of these MDS codes can be taken from $q$ to $\frac{n}{2}$.

In \cite{RefJ (2004) Grassl n=q^2} and \cite{RefJ (2008) Li n=q^2}, the authors constructed $[q^2,k]_{q^2}$ Hermitian self-orthogonal GRS codes, where $1\leq k\leq q-1$.
Here we consider larger dimensions.

\begin{lemma}\label{lem q22}
There exists a $[q^2,q]_{q^2}$ GRS code $GRS_q(\bm{a},\bm{v})$ such that
$$\begin{cases}
\bm{g}_{q-1}\bm{g}_{q-1}^\dagger\neq 0; \\
\bm{g}_i\bm{g}_j^\dagger=0, \ 0\leq i,\ j\leq q-1,\ except\ i=j=q-1.
\end{cases}$$
\end{lemma}

\begin{proof}
Label the elements of $\F_{q^2}$ as $\{a_1=0,a_2,\dots,a_{q^2}\}$.
Let $\bm{a}=(a_1,a_2,\dots,a_{q^2})\in \F_{q^2}^{q^2}$.
Denote $A=(\bm{a}^0,\bm{a}^1,\dots,\bm{a}^{q^2-2})^\mathcal{T}$ be an $(q^2-1)\times q^2$ matrix over $\F_{q^2}$.
It is easy to check that any $q^2-1$ columns of $A$ are linearly independent and $A^q$ is row equivalent to $A$.
By \cite[Lemma 2.3]{RefJ (2019)W. some},
we have that equation $A\bm{u}^T=\bm{0}^T$
has a solution $\bm{u}=(u_1,u_2,\dots,u_{q^2})\in (\F_{q}^*)^{q^2}$.
Let $\bm{v}=(v_1,v_2,\dots,v_{q^2})\in (\F_{q^2}^*)^{q^2}$ such that $\bm{v}^{q+1}=\bm{u}$.
It follows that $\bm{g}_i\bm{g}_j^\dagger=\langle \bm{a}^{qi+j},\bm{v}^{q+1}\rangle_0=0$ for
$1\leq i,j\leq q-1$ except $i=j=q-1$.
Moreover, it is easy to check that $\bm{g}_{q-1}\bm{g}_{q-1}^\dagger=-u_1\neq 0$. This completes the proof.
\end{proof}

\begin{theorem}\label{th MDS 1}
Let $n=q^2$.
For any $1\leq t\leq \lceil\frac{n}{2q}\rceil$ and
$$\begin{cases}
(t-1)q+1\leq k\leq tq-t, & l=k-(t-1)^2;\\
k=tq-t+1, &  l=k-(t-1)^2-1;\\
tq-t+2\leq k\leq tq, & l=2tq-t^2-k,
\end{cases}$$
there exists a $q^2$-ary $[n,k]$ MDS with $s$-dimensional Hermitian hull for each $0\leq s\leq l$.
\end{theorem}

\begin{proof}
Let $\bm{a}$, $\bm{v}$ be defined as Lemma \ref{lem q22},
then we have that
$$\begin{cases}
\bm{g}_{q-1}\bm{g}_{q-1}^\dagger\neq 0;\\
\bm{g}_i\bm{g}_j^\dagger=0,\ 0\leq i,\ j\leq q-1,\ except\ i=j=q-1.
\end{cases}$$
By Lemma \ref{lem q^2},
 the element in any place of matrix $G_{q^2}(\bm{a},\bm{v})G_{q^2}(\bm{a},\bm{v})^\dagger$ is one of $\bm{g}_i\bm{g}_j^\dagger$ for $0\leq i,j\leq q-1$.
 Note that for any $1\leq k\leq \lfloor\frac{n}{2}\rfloor$, $G_k(\bm{a},\bm{v})G_k(\bm{a},\bm{v})^\dagger$ is a $k$-order leading principal submatrix of $G_{q^2}(\bm{a},\bm{v})G_{q^2}(\bm{a},\bm{v})^\dagger$,
 then we only need to consider the position of $\bm{g}_{q-1}\bm{g}_{q-1}^{\dagger}$ in the matrix $G_k(\bm{a},\bm{v})G_k(\bm{a},\bm{v})^\dagger$.
 For $1\leq k\leq \lfloor\frac{n}{2}\rfloor$, let $\lambda_k$ be the number of elements $\bm{g}_{q-1}\bm{g}_{q-1}^\dagger$ in matrix $G_{k}(\bm{a},\bm{v})G_{k}(\bm{a},\bm{v})^\dagger$,
it follows that $rank(G_{k}(\bm{a},\bm{v})G_{k}(\bm{a},\bm{v})^\dagger)\leq \lambda_k$.
By Lemma \ref{lem q^2},
it is easy to check that
$$ \lambda_k=\begin{cases}
(t-1)^2, & for\ (t-1)q+1\leq k\leq tq-t ;\\
(t-1)^2+1, & for\  tq-t+1.\\
\end{cases}$$
Moreover, by Proposition \ref{pro 3},
for $tq-t+2\leq k\leq tq$,
we have
$$ \lambda_k\leq  (t-1)^2+1+2[k-(tq-t+1)]=t^2-2tq+2k.$$
The desired result follows from Lemmas \ref{lem Hull=GG} and \ref{lem 1}.
This completes the proof.
\end{proof}

\begin{remark}
For a GRS code $GRS_k(\bm{a},\bm{v})$, if $\Hull_H(GRS_q(\bm{a},\bm{v}))=GRS_{q-1}(\bm{a},\bm{v})$, then we get similar conclusions to those of Theorem \ref{th MDS 1}.
\end{remark}

Similarly, from \cite[Corollary 4.5]{RefJ (2017) Shi MDS} and \cite[Corollary 4]{RefJ (2023) Wan 1}, we can obtain the following lemma.

\begin{lemma}\label{lem h-1}
Let $n=\frac{(h-1)(q^2-1)}{h}$, where $h\mid (q+1)$ and $h\geq 2$, then there exists an $[n,q]_{q^2}$ GRS code $GRS_q(\bm{a},\bm{v})$ such that
$$\begin{cases}
\bm{g}_{q-1}\bm{g}_{\frac{q+1}{h}-2}^\dagger\neq 0;\\
\bm{g}_{\frac{q+1}{h}-2}\bm{g}_{q-1}^\dagger\neq 0;\\
\bm{g}_i\bm{g}_j^\dagger=0,\ 0\leq i,\ j\leq q-1,\ except\ i=q-1,\ j=\frac{q+1}{h}-2\ and\ i=\frac{q+1}{h}-2,\ j=q-1.
\end{cases}$$
\end{lemma}

\begin{proof}
Let $\bm{a}$, $\bm{v}$ be defined as \cite[Section 4.2]{RefJ (2023) Wan 1},
by \cite[Theorem 2]{RefJ (2023) Wan 1}, we have that $\bm{g}_i\bm{g}_j^{\dagger}=0$
for $0\leq i,j\leq q-2$.
Note that
$$\bm{g}_i\bm{g}_j^\dagger\neq 0\ {\rm for}\ i=q-1\ {\rm or}\ j=q-1\
{\rm if\ and\ only\ if}\ \frac{q^2-1}{h}\mid qi+j+q+1-\frac{q+1}{h}.$$
Similar to the proof of \cite[ Lemma 8]{RefJ (2023) Wan 1},
we have that when $h$ is even,
$$\bm{g}_i\bm{g}_j\neq 0\ {\rm for}\ i=q-1\ {\rm or}\ j=q-1\
{\rm if\ and\ only\ if}\  (i,j)=(q-1,\frac{q+1}{h}-2)\ {\rm or}\ (i,j)=(\frac{q+1}{h}-2,q-1).$$
Hence, when $h$ even, we have
$\bm{g}_i\bm{g}_j^\dagger=0$ for $0\leq i,j\leq q-1$ except $i=q-1$, $j=\frac{q+1}{h}-2$ and $i=\frac{q+1}{h}-2$, $j=q-1$.
Moreover, we know that $\bm{g}_{q-1}\bm{g}_{\frac{q+1}{h}-2}^\dagger\neq 0$ and $\bm{g}_{\frac{q+1}{h}-2}\bm{g}_{q-1}^\dagger\neq 0$.
Similarly, it can be shown that the result still holds when $h$ is odd. We omit the details.
This completes the proof.
\end{proof}

\begin{theorem}\label{th MDS 2}
Let $n=\frac{(h-1)(q^2-1)}{h}$, where $h\mid (q+1)$ and $h\geq 2$.
For any $1\leq t\leq \lceil\frac{n}{2q}\rceil$ and
$$\begin{cases}
(t-1)q+1\leq k\leq tq-t-\frac{(h-1)(q+1)}{h}+1, & l= k-2(t-1)^2;\\
tq-t-\frac{(h-1)(q+1)}{h}+2\leq k\leq tq-\frac{(h-1)(q+1)}{h}, &  l= 2tq+2t-2t^2-k-\frac{2(h-1)(q+1)}{h};\\
tq-\frac{(h-1)(q+1)}{h}+1\leq k\leq tq-t, & l=k+2t-2t^2;\\
tq-t+1\leq k\leq tq &  l= 2tq-2t^2-k,
\end{cases}$$
there exists a $q^2$-ary $[n,k]$ MDS with $s$-dimensional Hermitian hull for each $0\leq s\leq l$. \end{theorem}

\begin{proof}
Let $\bm{a}$, $\bm{v}$ be defined as Lemma \ref{lem h-1},
then we have that
 $$\bm{g}_i\bm{g}_j^\dagger=0\ {\rm for}\ 0\leq i,j\leq q-1\ {\rm except}\ i=q-1,\ j=\frac{q+1}{h}-2\ {\rm and}\ i=\frac{q+1}{h}-2,\ j=q-1$$
Similar to the proof of Theorem \ref{th MDS 1},
for $1\leq k\leq \lfloor\frac{n}{2}\rfloor$,
we only need to consider the position of $\bm{g}_{q-1}\bm{g}_{\frac{q+1}{h}+2}^{\dagger}$ and $\bm{g}_{\frac{q+1}{h}+2}\bm{g}_{q-1}^{\dagger}$ in the matrix $G_k(\bm{a},\bm{v})G_k(\bm{a},\bm{v})^\dagger$.
For $1\leq k\leq \lfloor\frac{n}{2}\rfloor$,
let $\lambda_k$ be the number of elements $\bm{g}_{q-1}\bm{g}_{\frac{q+1}{h}-2}^\dagger$ in matrix $G_{k}(\bm{a},\bm{v})G_{k}(\bm{a},\bm{v})^\dagger$.
By symmetry, we have that $rank(G_{k}(\bm{a},\bm{v})G_{k}(\bm{a},\bm{v})^\dagger)\leq 2\lambda_k$.
By Lemma \ref{lem q^2},
we can get
$$ \lambda_k=\begin{cases}
(t-1)^2, & for\ (t-1)q+1\leq k\leq tq-t-\frac{(h-1)(q+1)}{h}+1;\\
t^2-t-tq+k+\frac{(h-1)(q+1)}{h}, & for\  tq-t-\frac{(h-1)(q+1)}{h}+2\leq k\leq tq-\frac{(h-1)(q+1)}{h};\\
t^2-t, & for\  tq-\frac{(h-1)(q+1)}{h}+1\leq k\leq tq-t.\\
\end{cases}$$
Moreover, by Proposition \ref{pro 3},
for $tq-t+1\leq k\leq tq$,
we have
$$\lambda_k \leq t^2-t+k-(tq-t)=t^2-tq+k.$$
The desired result follows from Lemmas \ref{lem Hull=GG} and \ref{lem 1}.
 This completes the proof.
\end{proof}

Similarly, from \cite[Corollary 4]{RefJ (2023) Wan 2}, we can obtain the following lemma.

\begin{lemma}\label{lem 2h-1}
Let $n=\frac{(2h-1)(q^2-1)}{2h}$, where $q$ is odd and $\frac{q+1}{h}=2\tau+1$, for some $\tau\geq 1$.
Then there exists an $[n,q]_{q^2}$ GRS code $GRS_q(\bm{a},\bm{v})$ such that
$$\begin{cases}
\bm{g}_{q-1}\bm{g}_{\frac{q+1}{2}-2}^\dagger\neq 0;\\
\bm{g}_{\frac{q+1}{2}-2}\bm{g}_{q-1}^\dagger\neq 0;\\
\bm{g}_i\bm{g}_j^\dagger=0,\ 0\leq i,\ j\leq q-1,\ except\ i=q-1,\ j=\frac{q+1}{2}-2\ and\ i=\frac{q+1}{2}-2,\ j=q-1.
\end{cases}$$
\end{lemma}

\begin{proof}
Let $\bm{a}$, $\bm{v}$ be defined as \cite[Theorem 4]{RefJ (2023) Wan 2}, we have that $\bm{g}_i\bm{g}_j^{\dagger}=0$
for $0\leq i,j\leq q-2$.
Note that
$$\bm{g}_i\bm{g}_j^\dagger\neq 0\ {\rm for}\ i=q-1\ {\rm or}\ j=q-1\
{\rm if\ and\ only\ if}\ \frac{q^2-1}{2h}\mid qi+j+\frac{q+1}{2}.$$
Similar to the proof of \cite[Lemma 9]{RefJ (2023) Wan 2},
we have that when $h$ is even,
$$\bm{g}_i\bm{g}_j\neq 0\ {\rm for}\ i=q-1\ {\rm or}\ j=q-1\
{\rm if\ and\ only\ if}\  (i,j)=(q-1,\frac{q+1}{2}-2)\ {\rm or}\ (i,j)=(\frac{q+1}{2}-2,q-1).$$
The rest of the proof is similar to Lemma \ref{lem h-1}, and we omit the details.
This completes the proof.
\end{proof}

Similarly to Theorem \ref{th MDS 2}, by Lemma \ref{lem 2h-1}, we can directly obtain the following theorem.

\begin{theorem}\label{th MDS 3}
Let $n=\frac{(2h-1)(q^2-1)}{2h}$, where $q$ is odd and $\frac{q+1}{h}=2\tau+1$, for some $\tau\geq 1$.
For any $1\leq t\leq \lfloor \frac{n}{2q}\rfloor$ and
$$\begin{cases}
(t-1)q+1\leq k\leq tq-t-\frac{q+1}{2}+1, & l=k-2(t-1)^2;\\
tq-t-\frac{q+1}{2}+2\leq k\leq tq-\frac{q+1}{2}, &  l= 2tq+2t-2t^2-k-q-1;\\
tq-\frac{q+1}{2}+1\leq k\leq tq-t, & l= k+2t-2t^2;\\
tq-t+1\leq k\leq tq & l= 2tq-2t^2-k,
\end{cases}$$
there exists a $q^2$-ary $[n,k]$ MDS with $s$-dimensional Hermitian hull for each $0\leq s\leq l$.
\end{theorem}

\section{Application to EAQECCs and comparison of results}\label{sec4}

In this section,
we derive three new propagation rules for MDS EAQECCs constructed from GRS codes.
Then, from the presently known GRS codes with Hermitian hulls, we can directly obtain many MDS EAQECCs with more flexible parameters.
Moreover, based on the MDS codes we constructed before, we can obtain many new MDS EAQECCs with distances greater than $q+1$.
Finally, by comparison, we can show that the MDS EAQECCs we obtained have new parameters.

\subsection{Application to EAQECCs}
A $q$-ary quantum code of length $n$ and dimension $K$ is a $K$-dimension subspace of a $q^n$-dimensional Hilbert space $(\mathbb{C}^q)^{\otimes n}$.
A quantum code with minimum distance $d$ can detect any $d-1$ quantum errors and correct any $\lfloor \frac{d-1}{2}\rfloor$ quantum errors.
We use $[[n,k,d]]_q$ to denote a quantum code of length $n$, dimension $q^k$ and minimum distance $d$.
A quantum code attaining the bound
$$2d\leq n-k+2$$
is called a quantum MDS code.
In recent years, many quantum MDS codes with good parameters have been constructed from Hermitian self-orthogonal codes (see \cite{RefJ (2004) Grassl n=q^2}-\cite{RefJ (2023) Wan 2} and the references therein).

A $q$-ary EAQECC can be denoted as $[[n,k,d;c]]_q$ which encodes $k$ information qubits into
$n$ channel qubits with the help of $c$ pairs of maximally entangled states and corrects up
to $\lfloor \frac{d-1}{2}\rfloor$ errors, where $d$ is the minimum distance of the code.
For more information on EAQECCs, we refer the reader to \cite{RefJ (2012) Djordjevic}.
Similarly,
a tradeoff condition between the parameters $n$, $k$, $d$ and $c$ is the following bound:

\begin{lemma}(\cite{RefJ (2022) Grassl bound})
For an any EAQECC with parameters $[[n,k,d;c]]_q$,
\begin{align}
&k\leq c+\max\{0,n-2d+2\};\label{eqqqq 1}\\
&k\leq n-d+1;\label{eqqqq 2}\\
&k\leq \frac{(n-d+1)(c+2d-2-n)}{3d-3-n},\ if\ 2d\geq n+2.\label{eqqqq 3}
\end{align}
An EAQECC that satisfies Eq. (\ref{eqqqq 1}) for $2d\leq n+2$ or Eq. (\ref{eqqqq 3}) for $2d\geq n+2$ with equality is called an MDS EAQECC.
\end{lemma}

The next lemma shows that we can transform the construction of (MDS) EAQECCs into the construction of an MDS linear code with determined Hermitian hull dimensions.

\begin{lemma}\label{lem EAQECC}
(Hermitian construction \cite{RefJ (2018) K good})
Let $\C$ be an $[n,k,d]_{q^2}$ linear code and $\C^{\bot_H}$ be its Hermitian dual with parameters $[n,n-k,d^{\bot_H}]_{q^2}$.
Then there exist EAQECCs with parameters
$$[[n,k-dim(\Hull_H(\C)),d;n-k-dim(\Hull_H(\C))]]_q,$$ and
$$[[n,n-k-dim(\Hull_H(\C)),d^{\bot_H};k-dim(\Hull_H(\C))]]_q.$$
If $\C$ is MDS, then one of the two EAQECCs must be MDS.
\end{lemma}

By Lemma \ref{lem EAQECC} and Corollaries \ref{cor GRS 1}, \ref{cor GRS 2} and \ref{cor GRS 3}, we can directly obtain the following three new classes of propagation rules for MDS EAQECCs constructed from GRS codes.

\begin{theorem}\label{th rule 1} (MDS EAQECCs for longer length)
For $q>2$, if there exists an $[[n,n-k-l,k+1;k-l]]_q$ MDS EAQECC
constructed by Lemma \ref{lem EAQECC} from an $[n,k]_{q^2}$ GRS code with $l$-dimensional Hermitian hull, where $0\leq 2k\leq n$ and $0\leq l\leq k$.
Then for any integer $0\leq i\leq \min\{l,q^2+1-n\}$ and $0\leq s\leq l-i$, there exists an $[[n+i,n+i-k-s,k+1;k-s]]_q$ MDS EAQECC.
\end{theorem}

\begin{proof}
We apply the Hermitian construction in Lemma \ref{lem EAQECC} on the codes from Corollary \ref{cor GRS 1}. Then we can obtain the desired result. This completes the proof.
\end{proof}

\begin{theorem}\label{th rule 2} (MDS EAQECCs for larger distance)
For $q>2$, if there exists an $[[n,n-k-l,k+1;k-l]]_q$ MDS EAQECC
constructed by Lemma \ref{lem EAQECC} from an $[n,k]_{q^2}$ GRS code with $l$-dimensional Hermitian hull, where $0\leq 2k\leq n\leq q^2$ and $0\leq l\leq k$.
Then for any integer $0\leq i\leq \min \{l,\frac{n}{2}-k\}$ and $0\leq s\leq l-i$, there exists an $[[n,n-k-i-s,k+i+1;k+i-s]]_q$ MDS EAQECC.
\end{theorem}

\begin{proof}
We apply the Hermitian construction in Lemma \ref{lem EAQECC} on the codes from Corollary \ref{cor GRS 2}. Then we can obtain the desired result. This completes the proof.
\end{proof}

\begin{theorem}\label{th rule 3} (MDS EAQECCs for larger distance and longer length)
For $q>2$, if there exists an $[[n,n-k-l,k+1;k-l]]_q$ MDS EAQECC constructed by Lemma \ref{lem EAQECC} from an $[n,k]_{q^2}$ GRS code with $l$-dimensional Hermitian hull,
where $0\leq 2k\leq n$ and $0\leq l\leq k$.
Then for any integer $0\leq i\leq \min\{l,q^2+1-n,n-2k\}$ and $0\leq s\leq l-i$, there exists an $[[n+i,n-k-s,k+i+1;k+i-s]]_q$ MDS EAQECC.
\end{theorem}

\begin{proof}
We apply the Hermitian construction in Lemma \ref{lem EAQECC} on the codes from Corollary \ref{cor GRS 3}. Then we can obtain the desired result. This completes the proof.
\end{proof}

\newcommand{\tabincell}[2]{\begin{tabular}{@{}#1@{}}#2\end{tabular}}

\begin{table}
\caption{Some MDS EAQECCs constructed from an $[n,k]_{q^2}$ GRS code with $l$-dimensional Hermitian hull}
\label{tab:1}
\begin{center}
\resizebox{\textwidth}{20mm}{
	\begin{tabular}{c|c|c|c|c|c}
		\hline
		Class &Length  &  Distance $\leq \frac{n+2}{2}$ & entangled states &Range of $i$ & References\\
        \hline
       1& $n$
      & $k+1$ & $k-l\leq c\leq k$ & -& \cite{RefJ (2022) Chen IEEE,RefJ (2023) Chen} \\
       \hline
      2& $n-i$
      & $k+1$ & $k-l+i\leq c\leq k$ & $0\leq i\leq l$ & \cite{RefJ (2022) Luo rule} \\
       \hline
      3& $n-i$
      & $k-i+1$ & $k-l\leq c\leq k-i$ & $0\leq i\leq l$  & \cite{RefJ (2022) Luo rule} \\
       \hline
     4&  $n+i$
      & $k+1$ & $k-l+i\leq c\leq k$ & $0\leq i\leq \min\{l,q^2+1-n\}$  & Theorem \ref{th rule 1}\\
       \hline
      5& $n\leq q^2$
      & $k+i+1$ & $k-l+2i\leq c\leq k+i$ & $0\leq i\leq \min\{l,\frac{n}{2}-k\}$  & Theorem \ref{th rule 2}\\
       \hline
      6&  $n+i$
      & $k+i+1$ & $k-l+2i\leq c\leq k+i$ & $0\leq i\leq \min\{l,n-2k,q^2+1-n\}$  & Theorem \ref{th rule 3}\\
       \hline
	\end{tabular}}
 \begin{tablenotes}
     \footnotesize
    \item Note: in this table $0\leq l\leq k$.
    \end{tablenotes}
\end{center}
\end{table}

To illustrate the novelty of the propagation rules obtained in this paper. In Table \ref{tab:1}, we list the currently known propagation rules for MDS EAQECCs constructed from GRS codes.
From Table \ref{tab:1}, we can see that the MDS EAQECCs obtained from the propagation rules in this paper have larger lengths or larger distances, which is different from the currently known propagation rules. Therefore in this paper we obtain new propagation rules for MDS EAQECCs constructed from GRS codes.
Moreover, by Table \ref{tab:1}, we know that for a given GRS code with Hermitian hull,
we can get six classes of MDS EAQECCs with flexible parameters.
Then from the currently known GRS codes, we can directly obtain many MDS EAQECCs with more flexible parameters.

In the following, from the previously obtained MDS codes with Hermitian hulls, Table \ref{tab:1}, and Lemma \ref{lem EAQECC}, we obtain some (MDS) EAQECCs with flexible parameters.
It is worth noting that the distance of these MDS EAQECCs can be taken from $q+1$ to $\frac{n+2}{2}$,

\begin{theorem}\label{th q22}
Let $q>2$ be a prime power.
Let $n$, $k$ and $l$ be three positive integers.
If $n$, $k$ and $l$ satisfy one of the following three conditions:
\begin{itemize}
\item[(1)]
$n=q^2$, $\begin{cases}
(t-1)q+1\leq k\leq tq-t, & l=k-(t-1)^2;\\
k=tq-t+1, & l=k-(t-1)^2-1;\\
tq-t+2\leq k\leq tq, & l=2tq-t^2-k,
\end{cases}$ with $1\leq t\leq \lceil\frac{n}{2q}\rceil$;
\item[(2)]
$n=\frac{(h-1)(q^2-1)}{h}$, where $h\mid (q+1)$ and $h\geq 2$, and
$$\begin{cases}
(t-1)q+1\leq k\leq tq-t-\frac{(h-1)(q+1)}{h}+1, &  l= k-2(t-1)^2;\\
tq-t-\frac{(h-1)(q+1)}{h}+2\leq k\leq tq-\frac{(h-1)(q+1)}{h}, & l= 2tq+2t-2t^2-k-\frac{2(h-1)(q+1)}{h};\\
tq-\frac{(h-1)(q+1)}{h}+1\leq k\leq tq-t, & l= k+2t-2t^2;\\
tq-t+1\leq k\leq tq &  l=2tq-2t^2-k,
\end{cases}$$
 with $1\leq t\leq \lceil\frac{n}{2q}\rceil$;
 \item[(3)]
 $n=\frac{(2h-1)(q^2-1)}{2h}$, where $q$ is odd and $\frac{q+1}{h}=2\tau+1$, for some $\tau\geq 1$,
 and
$$\begin{cases}
(t-1)q+1\leq k\leq tq-t-\frac{q+1}{2}+1, &  l= k-2(t-1)^2;\\
tq-t-\frac{q+1}{2}+2\leq k\leq tq-\frac{q+1}{2}, &  l= 2tq+2t-2t^2-k-q-1;\\
tq-\frac{q+1}{2}+1\leq k\leq tq-t, &  l=k+2t-2t^2;\\
tq-t+1\leq k\leq tq &  l= 2tq-2t^2-k,
\end{cases}$$
 with $1\leq t\leq \lceil\frac{n}{2q}\rceil$,
\end{itemize}
then there exist (MDS) EAQECCs with parameters
$$\begin{cases}
[[n,n-k-i-s,k+i+1;k+i-s]]_q, & 0\leq i\leq l,\quad  0\leq s\leq l-i; \\
[[n,k+i-s,n-k-i+1;n-k-i-s]]_q,& 0\leq i\leq l,\quad   0\leq s\leq l-i;\\
[[n+i,n+i-k-s,k+1;k-s]]_q, & -l\leq i\leq \min\{l,q^2+1-n\},\ 0\leq s\leq l-|i|;\\
[[n+i,k-s,n+i-k+1,n+i-k-s]], & -l\leq i\leq \min\{l,q^2+1-n\},\ 0\leq s\leq l-|i|;\\
[[n+i,n-k-s,k+i+1;k+i-s]]_q, & -l\leq i\leq \min\{l,q^2+1-n\},\ 0\leq s\leq l-|i|;\\
[[n+i,k+i-s,n-k+1;n-k-s]]_q, & -l\leq i\leq \min\{l,q^2+1-n\},\ 0\leq s\leq l-|i|.
\end{cases}$$
\end{theorem}

In Table \ref{tab:MDS EAQECC}, we list some parameters of the MDS EAQECCs obtained from Theorem \ref{th q22}.
Moreover, in Tables \ref{tab:2}, \ref{tab:3} and \ref{tab:4}, we give some examples of Theorem \ref{th q22}.
From Tables \ref{tab:MDS EAQECC}, \ref{tab:2}, \ref{tab:3} and \ref{tab:4}, we can find that the MDS EAQECCs constructed in this paper have distances greater than $q+1$ and have flexible lengths and entangled states.

\begin{table}
\caption{Some of the parameters of the MDS EAQECCs we constructed.}
\label{tab:MDS EAQECC}
\begin{center}
\resizebox{\textwidth}{35mm}{
	\begin{tabular}{cccc}
		\hline
		 Length $n$  &  Distance $d\leq \frac{n+2}{2}$ & Entangled states $c$  & References\\
        \hline
        $n=q^2$ & $(t-1)q+2\leq d\leq tq-t+1$  & $(t-1)^2$ & Theorem \ref{th q22} (1)\\
              & $tq-t+2$  & $(t-1)^2+1$   &  \\
               & $tq-t+3\leq d\leq tq+1$  & $t^2-2tq+2d-2$    &  \\
        \hline
       \tabincell{c}{ $n=\frac{(h-1)(q^2-1)}{h}$, \\ $h\mid (q+1)$, $h\geq 2$} & $(t-1)q+2\leq d\leq tq-t-\frac{(h-1)(q+1)}{h}+2$ & $2(t-1)^2$& Theorem \ref{th q22} (2) \\
              & $tq-t-\frac{(h-1)(q+1)}{h}+3\leq d\leq tq-\frac{(h-1)(q+1)}{h}+1$  & $2(t^2-t-tq+d+\frac{(h-1)(q+1)}{h}-1)$    &  \\
               & $tq-\frac{(h-1)(q+1)}{h}+2\leq d\leq tq-t+1$  & $2t^2-2t$  &  \\
               & $tq-t+2\leq d\leq tq+1$  & $2(t^2-tq+d-1)$  &  \\
        \hline
        \tabincell{c}{ $n=\frac{(2h-1)(q^2-1)}{2h}$, $q$ odd,\\ $\frac{q+1}{h}=2\tau+1$, $\tau\geq 1$} & $(t-1)q+2\leq d\leq tq-t-\frac{q+1}{2}+2$ & $2(t-1)^2$& Theorem \ref{th q22} (3)\\
              & $tq-t-\frac{q+1}{2}+3\leq d\leq tq-\frac{q+1}{2}+1$  & $2t^2-2t-2tq+2d+q-1$    &  \\
               & $tq-\frac{q+1}{2}+2\leq d\leq tq-t+1$  & $2t^2-2t$  &  \\
               & $tq-t+2\leq d\leq tq+1$  & $2(t^2-tq+d-1)$  &  \\
        \hline
	\end{tabular}}
 \begin{tablenotes}
     \footnotesize
    \item Note: in this table $1\leq t\leq \lceil\frac{n}{2q}\rceil$.
    \end{tablenotes}
\end{center}
\end{table}

\begin{table}
\caption{Parameters of some known MDS EAQECCs with lengths $\leq q^2$ and minimum distances $>q+1$.}
\label{tab:EAQECCs}
\begin{center}
\resizebox{\textwidth}{98mm}{
	\begin{tabular}{ccccc}
		\hline
      Length $n$  &  Distance $d\leq \frac{n+2}{2}$ & Entangled states $c$ & Constraints & References\\
		\hline
		$2\lambda (q-1)$ & $d$  &  $2i$ &  \tabincell{c}{ $i\in \{1,2\}$, $\lambda$ is odd, $\lambda \mid (q-1)$,\\ $\frac{q-1}{2}(i-1)+4\lambda+1\leq d\leq \frac{q-1}{2}+2(i+1)\lambda$ \\ and $8\mid (q+1)$}& \cite{RefJ (2020) Lu}\\
         \hline
        $\frac{q^2+1}{\lambda}$ &  $\gamma$ & $\gamma-1$  & $2\leq \gamma\leq \lfloor \frac{q^2+1+2\lambda}{2\lambda}\rfloor$ & \cite{RefJ (2021) Sari}\\
        $\frac{q^2-1}{\lambda}$ &  $\gamma+1$ & $\gamma$  & $1\leq \gamma\leq \lfloor \frac{q^2-1}{2\lambda}\rfloor$ & \\
         \hline
         $q^2+i$ &  $d$ & $1$          & $2\leq d\leq 2q+i-1$ with $i\in \{-1,1\}$ & \cite{RefJ (2016) Fan q+1} \\
         $q^2$ &  $d$ & $1$            & $q+1\leq d\leq 2q-1$  &\\
         $\frac{q^2-1}{2}$ & $d$ & $2$  & $\frac{q+1}{2}+2\leq d\leq \frac{3}{2}q-\frac{1}{2}$  & \\
         $\frac{q^2-1}{t}$ & $d$ & $t$  &  \tabincell{c}{ $t\geq 3$ is odd with $t\mid (q+1)$ and
                                               \\  $\frac{(q+1)(t-1)+2t}{t}\leq d\leq \frac{(q+1)(t+1)-2t}{t}$}    &\\
        \hline
        $q^2+1$ &  $2q+2$ & $5$          & $q\equiv 3~({\rm mod}~4)$  & \cite{RefJ (2019) Sari} \\
        $q^2+1$ &  $2\lambda+2$ & $9$          & $q\equiv 3~({\rm mod}~4)$ and $q+1\leq \lambda\leq 2q-2$  & \\
        $\frac{q^2-1}{4}$ &  $d$ & $2$       & $q\equiv 3~({\rm mod}~4)$ and $\frac{q-9}{4}\leq d\leq  q$  & \\
        $\frac{q^2-1}{4}$ &  $d$ & $4$       & $q\equiv 3~({\rm mod}~4)$ and $\frac{3q+7}{4}\leq d\leq  \frac{5q-1}{4}$  & \\
        \hline
         $\frac{q^2-1}{8}$ &  $d$ & $8$          & $q+1\leq d\leq \frac{9q-1}{8}$  & \cite{RefJ (2019) Li} \\
         \hline
           $\frac{q^2-1}{i}$ &  $d$ & $i$          & $i\in \{3,5,7\}$ and $q+1\leq d\leq \frac{(i+1)(q+1)}{i}-1$  & \cite{RefJ (2018) Liu}\\
           $\frac{q^2-1}{j}$ &  $d$ & $j$          & $j\in \{5,7\}$ and $d=\frac{(j+1)(q+1)}{j}-1$  & \\
           $\frac{q^2-1}{7}$ &  $d$ & $7$          & $d=\frac{8(q+1)}{7}-1$  & \\
           $\frac{q^2-1}{4}$ &  $d$ & $4$          & $q+1\leq d\leq \frac{5(q+1)}{4}-1$  &\\
           $\frac{q^2-1}{6}$ &  $d$ & $6$          & $q+1\leq d\leq \frac{8(q+1)}{7}-1$  & \\
           \hline
            $\frac{q^2+1}{2}$ &  $3q-1$ & $13$          & $q\geq 5$  & \cite{RefJ (2020) Pang}\\
                              &  $4q-2$ & $25$          & $q\geq 7$  &\\
                              &  $6q-4$ & $61$          & $q\geq 11$  &\\
                              &  $7q-5$ & $85$          & $q\geq 13$  &\\
           \hline
           $i(q-1)$ &  $q+j+1$ & $q-t+2v+1$          & \tabincell{c}{ $q$ is a prime power, $1\leq j\leq\frac{q-3}{2}$, \\$i\equiv \pm 1~({\rm mod}~p)$ or $i\equiv \pm 0~({\rm mod}~p)$,\\
            $t=\gcd(i,q+1)$ and $v=\lfloor \frac{(j+1)t}{q+1}\rfloor$ }& \cite{RefJ (2024) Cheng} \\
             $2(q-1)$ &  $q+j+1$ & $q-1$          & $q=p^m$ is an odd prime power, $1\leq j\leq \frac{q-3}{2}$  & \\
             $3(q-1)$ &  $q+j+1$ & $q-t+2v+1$          & $m\geq 2$ and $v=\lfloor \frac{(j+1)t}{q+1}\rfloor$  & \\
             $4(q-1)$ &  $q+j+1$ & $q-t+2v+1$          & \tabincell{c}{$v=\lfloor \frac{(j+1)t}{q+1}\rfloor$, $m\geq 2$ when $p\equiv 1~({\rm mod}~4)$\\ and $m\neq 2$ when  $p\equiv 3~({\rm mod}~4)$}  & \\
               $i(q+1)$ &  $q+2$ & $q-s$         & \tabincell{c}{$q=p^m$ is a prime power, $p\mid i$,\\ $2\leq i\leq q-2$, $s=\gcd(i,q-1)$ } & \\
                $i(q+1)$ &  $q+2$ & $q-s+2$         & $p\nmid i$ & \\
                $i(q+1)$ &  $q+3$ & $q-s+2$         & $q=2^m$ & \\
           \hline
	\end{tabular}}
\end{center}
\end{table}

\begin{table}
\caption{Examples of MDS EAQECCs of Theorem \ref{th q22} (1) for $q=8$}
\label{tab:2}
\begin{center}
\resizebox{\textwidth}{26mm}{
	\begin{tabular}{cccc|cccc}
		\hline
		$k$ & $c$ & $s$ &  $[[n,n-2k+c,k+1;c]]_{q}$ &
        $k$ & $c$ & $s$ & $[[n-s,n-2k+s+c,k-s+1;c]]_{q}$ \\
        \hline

         $14$ & $1\leq c\leq 14$ & - & $[[64,36+c,15;c]]_{8}$ &
         $14$ & $1\leq c\leq 13$ & $1$& $[[63,37+c,14;c]]_{8}$ \\

         $21$ & $4\leq c\leq 21$ & - & $[[64,22+c,22;c]]_{8}$ &
         $21$ & $4\leq c\leq 18$ & $3$& $[[61,25+c,19;c]]_{8}$ \\

        $28$ & $8\leq c\leq 28$ & - & $[[64,8+c,29;c]]_{8}$ &
       $28$ & $8\leq c\leq 24$ & $4$ &$[[60,12+c,25;c]]_{8}$ \\

       $32$ & $16\leq c\leq 32$ & - & $[[64,c,33;c]]_{8}$ &
       $32$ & $16\leq c\leq 30$ & $2$& $[[62,2+c,31;c]]_{8}$ \\
       \hline
		$k$ & $c$ & $s$ & $[[n-s,n-2k-s+c,k+1;c]]_{q}$ &
        $k$ & $c$ & $s$ & $[[n+s,n-2k+s+c,k+1;c]]_{q}$ \\
        \hline
         $14$ & $2\leq c\leq 14$ & $1$ & $[[63,35+c,15;c]]_{8}$ &
         $14$ & $2\leq c\leq 14$ & $1$& $[[65,37+c,15;c]]_{8}$ \\

         $21$ & $6\leq c\leq 21$ & $2$ & $[[62,20+c,22;c]]_{8}$ &
         $21$ & $5\leq c\leq 21$ & $1$ & $[[65,23+c,22;c]]_{8}$ \\

         $28$ & $11\leq c\leq 28$ & $3$ & $[[61,5+c,29;c]]_{8}$ &
         $28$ & $9\leq c\leq 28$ & $1$& $[[65,9+c,29;c]]_{8}$ \\

       $31$ & $16\leq c\leq 31$ & $2$   &  $[[62,c,32;c]]_{8}$ &
       $32$ &  $17\leq c\leq 32$ &  $1$   & $[[65,1+c,33;c]]_{8}$ \\
       \hline

	\end{tabular}}
\end{center}
\end{table}

\begin{table}
\caption{Examples of MDS EAQECCs of Theorem \ref{th q22} (2) for $q=11$ and $h=3$}
\label{tab:3}
\begin{center}
\resizebox{\textwidth}{26mm}{
	\begin{tabular}{cccc|cccc}
		\hline
		$k$ & $c$ & $s$ &  $[[n,n-2k+c,k+1;c]]_{q}$ &
        $k$ & $c$ & $s$ & $[[n-s,n-2k+s+c,k-s+1;c]]_{q}$ \\
        \hline

         $13$ & $2\leq c\leq 13$ & - & $[[80,54+c,14;c]]_{11}$ &
         $13$ & $2\leq c\leq 11$ & $2$& $[[78,56+c,12;c]]_{11}$ \\

         $20$ & $4\leq c\leq 20$ & - & $[[80,40+c,21;c]]_{11}$ &
         $20$ & $4\leq c\leq 17$ & $3$ & $[[77,43+c,18;c]]_{11}$ \\

        $33$ & $18\leq c\leq 33$ & - & $[[80,14+c,34;c]]_{11}$ &
        $33$ & $18\leq c\leq 30$ & $3$ &$[[77,17+c,31;c]]_{11}$ \\

       $40$ & $24\leq c\leq 40$ & - &  $[[80,c,41;c]]_{11}$ &
       $40$ & $24\leq c\leq 39$ & $1$& $[[79,1+c,40;c]]_{11}$ \\
       \hline
		$k$ & $c$ & $s$ & $[[n-s,n-2k-s+c,k+1;c]]_{q}$ &
        $k$ & $c$ & $s$ & $[[n+s,n-2k+s+c,k+1;c]]_{q}$ \\
        \hline
         $13$ & $3\leq c\leq 13$ & $1$ & $[[79,53+c,14;c]]_{11}$ &
         $13$ & $3\leq c\leq 13$ & $1$& $[[81,55+c,14;c]]_{11}$ \\

         $20$ & $6\leq c\leq 20$ & $2$ & $[[78,38+c,21;c]]_{11}$ &
         $20$ & $6\leq c\leq 20$ & $2$ & $[[82,42+c,21;c]]_{11}$ \\

         $33$ & $21\leq c\leq 33$ & $3$ & $[[77,11+c,34;c]]_{11}$ &
         $33$ & $21\leq c\leq 33$ & $3$& $[[83,17+c,34;c]]_{11}$ \\

       $39$ & $26\leq c\leq 40$ & $2$   &  $[[78,c,40;c]]_{11}$ &
       $40$ &  $25\leq c\leq 40$ &  $1$   & $[[81,1+c,41;c]]_{11}$ \\
       \hline

	\end{tabular}}
\end{center}
\end{table}

\begin{table}
\caption{Examples of MDS EAQECCs of Theorem \ref{th q22} (3) for $q=9$ and $h=2$}
\label{tab:4}
\begin{center}
\resizebox{\textwidth}{26mm}{
	\begin{tabular}{cccc|cccc}
		\hline
		$k$ & $c$ & $s$ &  $[[n,n-2k+c,k+1;c]]_{q}$ &
        $k$ & $c$ & $s$ & $[[n-s,n-2k+s+c,k-s+1;c]]_{q}$ \\
        \hline

         $12$ & $2\leq c\leq 12$ & - & $[[60,36+c,13;c]]_{9}$ &
         $12$ & $2\leq c\leq 11$ & $1$& $[[59,37+c,12;c]]_{9}$ \\

         $16$ & $4\leq c\leq 16$ & - & $[[60,28+c,17;c]]_{9}$ &
         $16$ & $4\leq c\leq 14$ & $2$ & $[[58,30+c,15;c]]_{9}$ \\

         $24$ & $12\leq c\leq 24$ & -   & $[[60,12+c,25;c]]_{9}$ &
         $24$ & $12\leq c\leq 21$ & $3$& $[[57,15+c,22;c]]_{9}$ \\

        $30$ & $22\leq c\leq 30$ & - & $[[60,c,31;c]]_{9}$ &
        $30$ & $22\leq c\leq 28$ & $2$ &$[[58,2+c,29;c]]_{9}$ \\

       \hline
		$k$ & $c$ & $s$ & $[[n-s,n-2k-s+c,k+1;c]]_{q}$ &
        $k$ & $c$ & $s$ & $[[n+s,n-2k+s+c,k+1;c]]_{q}$ \\
        \hline
         $12$ & $3\leq c\leq 12$ & $1$ & $[[59,35+c,13;c]]_{9}$ &
         $12$ & $3\leq c\leq 12$ & $1$& $[[61,37+c,13;c]]_{9}$ \\

         $16$ & $4\leq c\leq 16$ & $2$ & $[[58,26+c,17;c]]_{9}$ &
         $16$ & $4\leq c\leq 16$ & $2$ & $[[62,30+c,17;c]]_{9}$ \\

         $24$ & $15\leq c\leq 24$ & $3$ & $[[57,9+c,25;c]]_{9}$ &
         $24$ & $15\leq c\leq 24$ & $3$& $[[63,15+c,25;c]]_{9}$ \\

       $29$ & $22\leq c\leq 29$ & $2$   &  $[[58,c,30;c]]_{9}$ &
       $30$ &  $23\leq c\leq 30$ &  $1$   & $[[61,1+c,31;c]]_{9}$ \\
       \hline

	\end{tabular}}
\end{center}
\end{table}

\subsection{Comparison of results}

 Recently, in \cite{RefJ (2024) Cheng}, Cheng et al. collated some of the currently known MDS EAQECCs
with lengths less than $q^2$ and minimum distances greater than $q+1$ (see Table \ref{tab:EAQECCs} or \cite[Table 1]{RefJ (2024) Cheng}).
From Tables \ref{tab:MDS EAQECC} and \ref{tab:EAQECCs}, we can find that the MDS EAQECCs obtained in this paper have different lengths, or have the same length but can be taken to larger distances.
Therefore, in this paper, MDS EAQECCs with new parameters are obtained.
It is worth noting that most of the MDS EAQECCs with length $n>q+1$ known currently have distances that cannot be taken from $q+1$ to $\frac{n+2}{2}$.
However, the distances of the MDS EAQECCs obtained in this paper can be taken from $q+1$ to $\frac{n+2}{2}$.
Moreover, using the propagation rules given in this paper, we can make the MDS EAQECCs recently given by Cheng et al. (\cite{RefJ (2024) Cheng}) have more flexible parameters to obtain new MDS EAQECCs.

\section{Conclusion}\label{sec6}

In this paper, we study the Hermitian hulls of GRS codes.
 For a given class of GRS codes, we obtain three classes of GRS codes with determined Hermitian hull dimensions (see Propositions \ref{pro GRS 1}, \ref{pro 3}, \ref{pro 4}).
Based on some known $q^2$-ary Hermitian self-orthogonal GRS codes with dimension $q-1$, by considering larger dimensions,
we obtain several classes of $q^2$-ary MDS codes with Hermitian hulls of arbitrary dimensions (see Theorems \ref{th MDS 1}, \ref{th MDS 2}, \ref{th MDS 3}).
It is worth noting that the dimension of these MDS codes can be taken from $q$ to $\frac{n}{2}$.
Furthermore, we derive three new propagation rules on MDS EAQECCs constructed from GRS codes (see Theorems \ref{th rule 1}, \ref{th rule 2}, \ref{th rule 3}).
 The propagation rules given in this paper can make the currently known MDS EAQECCs from GRS codes have more flexible parameters.
 Finally, we present several new classes of (MDS) EAQECCs with distances $d>q+1$ (see Theorem \ref{th q22}).

\section*{Acknowledgments}

This research was supported by the National Natural Science Foundation of China (No. U21A20428 and 12171134).

\section*{Data availability}

All data generated or analyzed during this study are included in this paper.

\section*{Conflict of Interest}

The authors declare that there is no possible conflict of interest.

\end{document}